\documentclass[11pt]{article}
\usepackage[a4paper, top=1in, bottom=1in, left=1in, right=1in]
{geometry}
\usepackage[T1]{fontenc}
\usepackage{graphicx}
\usepackage{xcolor}
\usepackage{tikz}
\usetikzlibrary{positioning,calc}
\usepackage{microtype}
\usepackage{amsthm,amsmath,amssymb,mathtools}
\usepackage{thmtools}
\usepackage{thm-restate}
\usepackage{hyperref}

\usepackage[noabbrev]{cleveref}
\usepackage{booktabs}
\usepackage{enumitem}
\usepackage{algorithm}
\usepackage{algpseudocode}
\usepackage{float}

\renewcommand{\epsilon}{\varepsilon}

\newtheorem{definition}{Definition}[section]

\newcommand{\cF}{\mathcal{F}}
\newcommand{\cO}{\mathcal{O}}
\newcommand{\cG}{\mathcal{G}}
\newcommand{\SW}{\mathsf{SW}}

\usepackage{etoolbox}

\floatstyle{ruled}
\newfloat{subroutine}{htbp}{los}
\floatname{subroutine}{Subroutine}

\title{Settling PROPm and PROPavg in Graphical Resource Allocation}

\author{Bo Li\thanks{Department of Computing, The Hong Kong Polytechnic University, China. comp-bo.li@polyu.edu.hk} \and
Ankang Sun\thanks{School of Computer Science and Technology, Shandong University, Qingdao, China. ankang.sun@sdu.edu.cn} \and
Ruijie Wang\thanks{EECS, University of Michigan, Ann Arbor, MI, USA. wruijie@umich.edu}}

\begin{document}
\date{}
\maketitle

\begin{abstract}
We study proportional fairness in graphical resource allocation, where agents are vertices, indivisible items are edges, and each item must be allocated to one of its two endpoints. It has been proved that PROP1 orientations always exist and PROPx orientations may not, but it has remained open whether the intermediate relaxations PROPm and PROPavg (both can be satisfied without graphical constraints) can always be satisfied.
In this paper, we resolve this gap. 
We prove that PROPm orientations always exist for goods on multigraphs and can be computed efficiently. The guarantee extends to chores and a mixture of goods and chores. In sharp contrast, we show that PROPavg orientations need not exist, even for simple graphs with binary valuations, and that deciding their existence is NP-complete. We further quantify the efficiency loss of PROPm: for goods, the price of PROPm is exactly $2$, which is one of the few settings where a constant bound on the price of fairness can be obtained. We complement these results with hardness results for welfare optimization under PROPm and for the existence of orientations that are simultaneously PROPm and Pareto optimal.
\end{abstract}

\section{Introduction}

Fair division of indivisible resources has been extensively studied across a wide range of settings and fairness notions \cite{DBLP:journals/ai/AmanatidisABFLMVW23}. One of the most fundamental fairness requirements is proportionality (PROP), which requires each agent to receive utility at least a $\frac{1}{n}$-fraction of her value for all items, where $n$ is the number of agents. Since exact proportionality is often unachievably demanding for indivisible resources, several relaxations have been proposed. These relaxations typically allow an agent to compensate for a proportionality shortfall by hypothetically adding back one carefully chosen item. The weakest among them is proportionality up to one item (PROP1) \cite{DBLP:conf/sigecom/ConitzerF017}, which requires that each agent can reach her proportional share after the addition of some item allocated to another agent. PROP1 allocations are guaranteed to exist and can be achieved together with Pareto optimality (PO) \cite{Caragiannis2019unreasonable,DBLP:conf/aaai/BarmanK19}. This naturally raises the question of whether stronger relaxations of proportionality can also be guaranteed.

Baklanov et al. \cite{DBLP:conf/aaai/BaklanovGGS21} introduced proportionality up to the maximin item (PROPm), where the hypothetical item added to an agent's bundle is determined by the maximum, over all other agents, of the least-valued item in that agent's bundle. They first proved the existence of PROPm allocations for up to five agents \cite{DBLP:conf/aaai/BaklanovGGS21} and later extended the result to an arbitrary number of agents \cite{DBLP:conf/ijcai/BaklanovGG021}. Subsequently, Kobayashi and Mahara \cite{DBLP:journals/siamdm/KobayashiM25} introduced a stronger notion, proportionality up to the least-valued item on average (PROPavg), in which the added value is the average value of the least-valued items in the bundles of the other agents. They showed that PROPavg allocations are also guaranteed to exist for all instances.
It is worth noting that if the hypothetical compensation is further restricted to the least-valued item among all items allocated to the other agents, i.e., PROPx, existence is no longer guaranteed \cite{moulin2019fair,DBLP:journals/orl/AzizMS20}. We summarize the following hierarchy of these notions by implication:
\[
\text{PROPx}\Longrightarrow \text{PROPavg}\Longrightarrow \text{PROPm}\Longrightarrow \text{PROP1}.
\]

In all of the models discussed above, indivisible resources are allocated in a fully flexible manner: each item may be assigned to any agent. In many applications, however, feasibility is inherently local. For example, a course may be taught only by one of a few qualified instructors, a task may be handled only by one of the adjacent teams, and an infrastructure resource may be shared only by the two parties it connects. A natural abstraction for such settings is a graph: agents are represented by vertices, items by edges, and allocating an item corresponds to orienting the associated edge toward one of its endpoints.

This graphical allocation model is closely related to graph balancing \cite{DBLP:conf/soda/EbenlendrKS08,DBLP:journals/algorithmica/EbenlendrKS14} and to fair division under graphical valuations \cite{DBLP:conf/sigecom/0001FKS23,christodoulou2026envy}, in particular its orientation variant. It captures the tension between local feasibility constraints and global fairness requirements and has been widely studied at the interface of theoretical computer science and economics. The graph-balancing literature mainly focuses on computational questions related to maximizing the minimum utility \cite{DBLP:conf/focs/ChakrabartyCK09,DBLP:conf/stoc/BateniCG09,DBLP:journals/scheduling/VerschaeW14,DBLP:conf/icalp/JansenR19,DBLP:conf/esa/SchwartzY19}, whereas the fair-orientation literature studies fairness notions such as envy-freeness and its relaxations \cite{DBLP:conf/sigecom/0001FKS23,Zeng2025efx,DBLP:conf/ijcai/DeligkasEGK25,DBLP:conf/ijcai/0027WL024,hsu2025polynomial}.
Very recently, Sun et al. \cite{DBLP:journals/corr/abs-2602-18098} initiated the study of proportionality-based fairness in graphical resource allocation. In this model, since each edge is available to exactly two possible recipients, the proportional benchmark is naturally defined as one half of an agent's total value for her incident edges \cite{DBLP:journals/corr/abs-2602-18098}. Sun et al. showed that PROP1 allocations are always achievable, whereas PROPx allocations are not guaranteed to exist. They explicitly left open the intermediate existence of PROPm.

In this paper, we resolve this open problem and provide a systematic study of PROPm and PROPavg in graphical resource allocation, covering their existence, computational complexity, and efficiency guarantees.

\subsection{Our Contribution}

In this paper, we study proportional fairness in graphical resource allocation. Following \cite{DBLP:journals/corr/abs-2602-18098}, we define the proportional share of each agent $i$ as
\[
    \mathsf{PROP}_i=\frac{1}{2} v_i(E_i),
\]
where $v_i(\cdot)$ is agent $i$'s valuation function and $E_i$ is the set of edges incident to $i$. We adopt the notions of PROPm and PROPavg from \cite{DBLP:conf/aaai/BaklanovGGS21} and \cite{DBLP:journals/siamdm/KobayashiM25}, respectively. 
We give a negative answer for PROPavg and a positive answer for PROPm. Our results show a sharp boundary: PROPm is robust in graphical resource allocation, while PROPavg is already too strong even on simple graphs with binary valuations.
For PROPavg, we prove that deciding whether a PROPavg allocation exists or not is NP-complete.
For PROPm, we conduct a complete analysis, which is summarized as follows.

\paragraph{Existence of PROPm Orientations.}
Our first main result is that PROPm orientations always exist for goods on multigraphs and can be computed in polynomial time. The algorithm first resolves pairs of adjacent agents when a local sufficient condition guarantees that both endpoints receive at least their local half-share. The unresolved pairs are then compressed into a simple virtual graph whose endpoint weights record the remaining local deficits and surpluses. A greedy-walk orientation of this virtual graph guarantees that, for every agent, all but possibly one lost virtual edge can be charged to previously won virtual edges. Mapping the virtual orientation back to the original multigraph ensures a PROPm orientation. The algorithm runs in $O(m\log m)$ time.
The existence guarantee extends beyond goods: we prove that a PROPm orientation also exists for chores and a mixture of goods and chores.


\paragraph{The Price of PROPm.}
Fairness and efficiency are often orthogonal criteria, so we next quantify the welfare loss caused by imposing PROPm. We use the standard price-of-fairness benchmark: the worst-case ratio between the maximum social welfare without the fairness constraint and the maximum social welfare among PROPm orientations. For indivisible goods in graphical resource allocation, we prove that the price of PROPm is exactly $2$.
This constant bound is notable in light of known price-of-fairness results for unconstrained indivisible goods. Without graphical constraints, the price of PROP1 can be as large as $\Omega(\sqrt n)$, and the same order of loss also appears for other fairness notions such as EF1 and $\frac{1}{2}$-MMS \cite{DBLP:conf/wine/BarmanB020,DBLP:journals/mst/BeiLMS21}. 
We give a constructive proof. We construct a PROPm orientation such that reversing the assignments of all contested edges yields another PROPm orientation.
An averaging argument then shows that one of them gives at least half of the social welfare. 

\paragraph{Optimization under PROPm.}
We complement the existence and price-of-fairness results by studying the computational complexity of welfare optimization under PROPm. We prove that computing a PROPm orientation with maximum social welfare is APX-hard on multigraphs and NP-hard even on simple graphs. We also show that PROPm is not generally compatible with Pareto optimality in the graphical setting and that deciding whether there exists an orientation that is simultaneously PROPm and Pareto optimal is NP-hard. These results extend, via a corresponding goods-chores duality, to graphical chores.

In the main body of the paper, we only focus on the orientation of goods and defer the study of chores and the mixture of goods and chores to Appendix \ref{sec:chores-mixed}.

\subsection{Related Work}

\paragraph{Max-min and Min-max Graph Balancing.}
Graphical allocation is also closely related to graph balancing and restricted allocation with max-min or min-max objectives. In these problems, each job or good can be assigned only to a limited set of machines or agents, and the special case in which each item is feasible for at most two has the same combinatorial flavor as orienting graph edges. The Santa Claus problem \cite{DBLP:conf/stoc/BansalS06,DBLP:journals/talg/AsadpourFS12,DBLP:conf/icalp/ChengM18} and the (unrelated) job scheduling problem \cite{DBLP:journals/mp/LenstraST90,DBLP:journals/mp/ShmoysT93} have been studied extensively. 
For the min-max objective, it was shown in \cite{DBLP:conf/soda/EbenlendrKS08,DBLP:journals/algorithmica/EbenlendrKS14} and \cite{DBLP:journals/jco/AsahiroJMOZ11} that no approximation ratio better than $1.5$ is possible unless $\mathrm{P}=\mathrm{NP}$. Ebenlendr et al. \cite{DBLP:conf/soda/EbenlendrKS08,DBLP:journals/algorithmica/EbenlendrKS14} gave a polynomial-time $1.75$-approximation algorithm, which was later slightly improved by Jansen and Rohwedder \cite{DBLP:conf/icalp/JansenR19}.
For the max-min objective, Chakrabarty et al. \cite{DBLP:conf/focs/ChakrabartyCK09} proved that the problem is hard to approximate within any factor better than $2$. On the positive side, polynomial-time approximation algorithms with ratios $4$, $2+\epsilon$, and $2$ were developed by Bateni et al. \cite{DBLP:conf/stoc/BateniCG09}, Chakrabarty et al. \cite{DBLP:conf/focs/ChakrabartyCK09}, and Verschae and Wiese \cite{DBLP:journals/scheduling/VerschaeW14}, respectively. A variant incorporating orientation costs was studied in \cite{DBLP:conf/esa/SchwartzY19}. More general restricted-valuation models, where each item may be valuable to an arbitrary subset of agents, have also been investigated, for example in \cite{DBLP:journals/talg/AsadpourFS12,DBLP:conf/icalp/ChengM18,DBLP:conf/soda/Feige08}.

\paragraph{Fair allocations and orientations in graphs.}
Our model is the orientation variant of the graphical valuation model introduced by Christodoulou et al.~\cite{DBLP:conf/sigecom/0001FKS23,christodoulou2026envy}, in which agents are vertices, items are edges, and each agent values only the edges incident to her. They distinguish \emph{allocations}, where an edge may be given to any agent (including one who does not value it), from \emph{orientations}, where every edge must be given to one of its endpoints, so that no item is wasted. They showed that EFX allocations always exist on simple graphs, whereas EFX orientations may fail to exist and deciding whether one exists is NP-complete. The two variants have since been studied along largely separate lines.

On the allocation side, the central question has been whether EFX allocations also exist on multigraphs. Amanatidis et al.~\cite{DBLP:conf/sigecom/AmanatidisFS24} proved that $\frac{2}{3}$-EFX allocations exist on multigraphs with additive valuations, which Kaviani et al.~\cite{kaviani2025improvedapproximateefxguarantees} improved to $\frac{\sqrt{2}}{2}$; Kaviani et al.~\cite{kaviani2024tale} showed that exact EFX allocations exist for restricted additive valuations. Exact EFX allocations were then established for bipartite multigraphs \cite{afshinmehr2024efx,bhaskar2025extending,sgouritsa2025existence}, multicycles \cite{afshinmehr2024efx}, multitrees and multigraphs whose girth is large relative to their chromatic number \cite{bhaskar2025extending}, multigraphs whose shortest cycle of non-parallel edges has length at least 6 \cite{sgouritsa2025existence}, and triangle-free multigraphs \cite{afshinmehr2026trianglefree}. 
Recently, Afshinmehr et al.~\cite{afshinmehr2026multigraphs} and Christodoulou et al.~\cite{christodoulou2026efxmulti} resolved the general case, showing that EFX allocations always exist on multigraphs with cancelable valuations.
The model has also been extended to (multi-)hypergraphs \cite{lianeas2026efx,kakatelis2026almost}. 
Misra and Sethia~\cite{misra2024envy} studied envy-free and efficient allocations under graphical valuations. 
Beyond envy-based notions, Christodoulou and Mastrakoulis~\cite{christodoulou2025exact} showed that exact MMS allocations exist on multigraphs with additive valuations (for MMS, one may restrict attention to orientations without loss of generality), and separated PMMS allocations from PMMS orientations.

On the orientation side, which is the focus of this paper, Zeng and Mehta~\cite{Zeng2025efx} related EFX orientability of simple graphs to the chromatic number: every graph with chromatic number at most 2 admits an EFX orientation under any valuations, whereas this fails for chromatic number at least 4. 
Zeng~\cite{zeng2026polynomial} subsequently gave a complete characterization of such strongly EFX orientable graphs, which can be checked in polynomial time.
Deligkas et al.~\cite{DBLP:conf/ijcai/DeligkasEGK25} proved that EF1 orientations always exist for monotone valuations, and that deciding the existence of an EFX orientation remains NP-complete even on graphs with vertex cover number $8$ or on multigraphs with $10$ vertices.
Bla{\v z}ej et al.~\cite{Blavzej2025tractable} and Kanellopoulos et al.~\cite{DBLP:journals/corr/abs-2512-25033} studied tractable graph structures and the parameterized complexity of EFX and EF orientations; the latter strengthened the hardness to vertex cover number $4$ and to multigraphs with $4$ vertices. Filos-Ratsikas et al.~\cite{filosratsikas2026efkx} showed that even EF$k$X orientations may not exist and that deciding their existence is NP-complete.
For multigraphs, Afshinmehr et al.~\cite{afshinmehr2024efx} characterized the existence of EFX orientations on bipartite multigraphs, and Hsu~\cite{hsu2025efxmulti} identified non-trivial odd multitrees as the obstruction for bi-valued symmetric multigraphs. 
For chores and mixed manna, Zhou et al.~\cite{DBLP:conf/ijcai/0027WL024} gave a complete picture of EFX allocations and orientations on simple graphs.
Zhou et al.~\cite{DBLP:conf/ijcai/0027WL024}, Sun et al.~\cite{DBLP:journals/corr/abs-2602-18098}, and Hsu and King~\cite{hsu2025polynomial} showed that EF1 and EFX orientations of chores can be found in polynomial time on simple graphs whenever they exist, while the corresponding decision problems are NP-complete on multigraphs.
Li et al.~\cite{li2025subsidy,li2026eliminating} studied the minimum subsidy for EF orientations, showing that a total subsidy of $\lfloor n/2\rfloor$ suffices on multigraphs with additive valuations. 
Li et al.~\cite{li2026efficient} generalized the orientation model for chores to position-aware valuations, where vertices are positions to which agents are assigned, and studied the complexity and approximability of Pareto optimal, cost-minimizing, and MMS assignments. 
Biburger et al.~\cite{biburger2026online} studied the approximability of online fair orientations.

Finally, graphs play other roles in fair division. 
A graph may restrict which agents can envy each other \cite{payan2023relaxations}, or require each agent's bundle to be connected in a graph over the items \cite{bouveret2017fair}.



\section{Preliminaries}
For a positive integer $k$, let $[k]=\{1,\ldots,k\}$. We study fair orientations of a multigraph $G=(N,E)$, where the vertices are agents $N = [n]$, and the edges are indivisible items. Write $m=|E|$. 
For agents $i,j\in N$, let $E_{ij}$ denote the set of edges between $i$ and $j$. In a simple graph, $|E_{ij}|\le 1$. Let $E_i$ be the set of edges incident to agent $i$, and let $N_i$ be the set of neighbors of $i$.
An orientation $\pi=(\pi_i)_{i\in N}$ assigns each edge to one of its endpoints. Thus $\pi_i\subseteq E_i$ is the set of edges assigned to agent $i$, and the sets $\pi_i$ form a partition of $E$. Each agent $i$ has an additive valuation $v_i(\cdot) : 2^E \to \mathbb R$ over incident edges, so for every $S\subseteq E_i$, $v_i(S)=\sum_{e\in S}v_i(e)$. 
We write an instance as $I=(G, v)$, where $v=( v_1, \ldots, v_n)$ and $v_i = (v_i(e))_{e \in E_i}$ for $i \in [n]$. 
Most of this paper concerns goods instances, in which $v_i(e)\ge 0$ for every agent $i$ and every edge $e\in E_i$. To keep the preliminaries focused, we state the solution concepts for goods here and defer the corresponding definitions for chores and mixed items to Appendix~\ref{sec:chores-mixed}.

For an agent $i$, since each edge is available to two endpoints, its proportional share is given as $\mathsf{PROP}_i=\frac{1}{2} v_i(E_i)$.
For an orientation $\pi$ and an agent $i$, define
\[
    N^{\pi}_+(i)=\{j\in N_i: E_{ij}\cap \pi_j\neq \emptyset\}
    \quad\text{and}\quad
    N^{\pi}_-(i)=\{j\in N_i: E_{ij}\cap \pi_i\neq \emptyset\}.
\]
Thus $N^{\pi}_+(i)$ is the set of neighbors to whom $i$ loses at least one edge, and $N^{\pi}_-(i)$ is the set of neighbors from whom $i$ receives at least one edge under $\pi$.

For fairness notions, we first define PROPm, which asks an agent to either receive her proportional share or reach it after recovering one lost good. When several parallel edges are lost to the same neighbor, the guarantee uses the least valuable such edge and then chooses the best neighbor.
\begin{definition}[PROPm]\label{def:propm}
An orientation $\pi$ of a goods instance is \emph{PROPm} if, for every agent
$i\in N$, either $N^{\pi}_+(i)=\emptyset$, or $v_i(\pi_i)+ \max_{j\in N^{\pi}_+(i)} \min_{e\in E_{ij}\cap \pi_j}v_i(e) \geq \mathsf{PROP}_i$.
\end{definition}

If $N^{\pi}_+(i)=\emptyset$, then agent $i$ receives all incident edges and has a value of at least her proportional share.
When $N^{\pi}_+(i)\neq\emptyset$, refer to the second term on the left-hand side of the inequality in Definition~\ref{def:propm} as the \emph{correction} of agent $i$ under $\pi$.
PROPavg also adapts the local correction, but it replaces the best-neighbor correction in PROPm with the average correction across all neighbors to whom the agent loses edges. 
Since the average-minimum correction of PROPavg is no larger than the maximin correction of PROPm, PROPavg implies PROPm.

\begin{definition}[PROPavg]\label{def:prop-avg}
An orientation $\pi$ of a goods instance is \emph{PROPavg} if, for every agent
$i\in N$, either $N^{\pi}_+(i)=\emptyset$, or
$ v_i(\pi_i)+ \frac{1}{|N^{\pi}_+(i)|}\sum_{j\in N^{\pi}_+(i)} \min_{e\in E_{ij}\cap \pi_j}v_i(e) \geq \mathsf{PROP}_i$.
\end{definition}

For the unconstrained PROPm or PROPavg allocation, correction for agent $i$ is defined by evaluating the minimum value over the entire bundle $\pi_j$ of another agent $j$. 
However, applying this to graphical orientations is problematic. Because agent $i$ only values her incident edges, if a neighbor $j$ receives any edge not incident to $i$, the minimum value $v_i(e)$ over $e \in \pi_j$ would trivially be zero. 
This would cause the PROPm and PROPavg corrections to almost always collapse to zero and reduce these relaxations to PROP. 
Therefore, to adapt these relaxations, we restrict the evaluation of a neighbor's bundle to the edges that agent $i$ actually competed for and lost to that neighbor, i.e., $E_{ij} \cap \pi_j$.

It is also worth noting that, though the $\frac{1}{2}$-share required by PROPm and PROPavg orientations is no less than the $\frac{1}{n}$-share required by the unconstrained PROPm and PROPavg allocation, PROPm orientations do not imply PROPm allocations even assuming agents have zero value for non-incident edges. The same holds for PROPavg.
Two examples show that the notions are incomparable. In both examples, all unspecified endpoint values are zero. 
First, an orientation can be PROPavg (thus PROPm) without being a PROPm allocation. Let $N=\{1,2,3\}$. Edges $a=(1,2)$ and $b=(1,3)$ have $v_1(a)=v_1(b)=1$; $c, d \in E_{23}$ are dummy edges. 
Consider the orientation $\pi_1=\emptyset$, $\pi_2=\{a,c\}$, and $\pi_3=\{b,d\}$. 
Agent 1's value plus correction is $0 + \frac{1}{2}(1 + 1) = 1$, which meets $\mathsf{PROP}_1  = 1$. 
However, the unrestricted PROPm evaluation includes the dummy edges, yielding $0 + \max(\min(1,0), \min(1,0)) = 0$.
This fails the $\frac{1}{n}$-share of $\frac{2}{3}$.
Conversely, an orientation can be a PROPavg (thus PROPm) allocation but not a PROPm orientation. 
Let $N=\{1,2,3\}$. Edges $a,b,c \in E_{12}$ are each worth $5$ to agent 1, $d \in E_{13}$ is worth $4$, and $f \in E_{23}$ is a dummy edge. 
Consider the allocation $\pi_1=\{d\}$, $\pi_2=\{a,b,c\}$, and $\pi_3=\{f\}$. 
For agent 1, the unrestricted PROPavg evaluation is $4 + \frac{5+0}{2} = \frac{13}{2}$, which exceeds the $\frac{1}{n}$-share of $\frac{19}{3}$. 
However, under Definition~\ref{def:propm}, agent 1's orientation $\text{PROPm}$ share is $\mathsf{PROP}_1 = \frac{19}{2}$. Her orientation evaluation remains $4 + 5 = 9$, but since $9 < \frac{19}{2}$, it is not a PROPm orientation.

We next define social welfare and Pareto optimality to study the efficiency of orientations.

\begin{definition}[Social welfare]
    For an orientation $\pi$, the social welfare is $\mathsf{SW}(\pi)=\sum_{i\in N}v_i(\pi_i)$.
\end{definition}

\begin{definition}[Pareto optimal]
    An orientation $\pi'$ \emph{Pareto dominates} an orientation $\pi$ if $v_i(\pi'_i)\ge v_i(\pi_i)$ for every agent $i\in N$, and $v_j(\pi'_j)>v_j(\pi_j)$ for at least one agent $j\in N$. 
    An orientation $\pi$ is \emph{Pareto optimal} (PO) if no orientation Pareto dominates it.
\end{definition}

We study the interaction between fairness and efficiency in two ways: (1) we ask how much social welfare would be sacrificed under fairness constraints, and (2) we ask whether there exists an orientation that is both fair and Pareto optimal.
In order to quantify the welfare loss under fairness constraints, we apply the framework of the \emph{price of fairness}. Intuitively, it is the supremum ratio over all instances between the unconstrained optimum welfare and the optimal welfare of all fair allocations.

\begin{definition}[Price of fairness]
    The price of fairness with respect to the fairness notion $\cF$ is
    \[
    \text{Price of } \cF = \sup\limits_{I} \frac{\max_{\pi\in \cO(I)} \SW(\pi)}{\max_{\pi\in \cO_{\cF}(I)}\SW(\pi)},
    \]
    where $\cO(I)$ is the set of all orientations and $\cO_{\cF}(I)$ is the set of $\cF$ orientations. The supremum ranges over goods instances with positive optimum welfare and a nonempty $\cO_{\cF}(I)$; a zero denominator gives ratio $+\infty$.
\end{definition}

\section{Existence and Computation of PROPm Orientations for Goods} \label{sec:PROPmExists}

We begin with PROPm orientations for goods. Our main result in this section is that, in the allocation of goods, PROPm orientations always exist and can be found efficiently.

\begin{restatable}{theorem}{thmPropmExists}\label{thm:propmExists}
For goods and multigraphs, one can compute PROPm orientations in polynomial time.
\end{restatable}

The formal introduction of the algorithm is presented in \textsc{Compute-PROPm} (Algorithm~\ref{alg:main}). At a high level, the algorithm works as follows:  it computes a PROPm orientation by first resolving easy sets of parallel edges locally, then handling the remaining sets of edges via a virtual graph. 

For two adjacent agents $i$ and $j$, \textsc{Min-Majority-Partition} (Subroutine~\ref{sub:min_majority}) finds an inclusion-minimal set $P_{ij}\subseteq E_{ij}$ that gives one endpoint, say $i$, at least half of
its value for $E_{ij}$. Then \textsc{Pairwise-Resolve} (Subroutine~\ref{sub:pairwise}) tests sufficient conditions for allocating $E_{ij}$ so that both endpoints receive at least their half-share. When a condition holds, all edges in $E_{ij}$ are oriented immediately
and removed from further consideration.

The remaining edge sets are handled by \textsc{Build-Virtual-Graph} (Subroutine~\ref{sub:virtual_graph}). Each
unresolved edge set $E_{ij}$ is represented by one virtual edge $\{u_i,u_j\}$.
Its weight at endpoint $u_i$ is
\[
    w_i(\{u_i,u_j\})
    =
    2\left(v_i(P_{ij})-\frac{1}{2} v_i(E_{ij})\right).
\]
This weight reflects the local imbalance of giving $P_{ij}$ to $i$ rather than $E_{ij}\setminus P_{ij}$ to $i$. 
If the virtual edge is oriented toward $u_i$, then $i$ receives $P_{ij}$ and gains surplus $\frac{1}{2}w_i(\{u_i,u_j\})$ on $E_{ij}$.
If it is oriented away from $u_i$, then $i$ receives $E_{ij}\setminus P_{ij}$ and incurs deficit $\frac{1}{2}w_i(\{u_i,u_j\})$ on $E_{ij}$.

Finally, \textsc{Greedy-Walk} (Subroutine~\ref{sub:greedy_walk}) orients the virtual graph so that, for each agent, the total deficit from unreceived virtual edges is covered by the surplus from received virtual edges, except possibly for the first unreceived edge. 
After converting the virtual orientation back to the original multigraph, the remaining unreceived edge corresponds to a neighbor whose received bundle provides the PROPm correction. Thus, every agent either receives her proportional share or satisfies PROPm using one allowed maximin lost bundle.

Throughout this section, for every unordered adjacent pair $\{i,j\}$ we fix one of the two endpoints and denote it by $i$; the other endpoint is denoted by $j$. This convention is used only for naming the set $P_{ij}$.
For a set of parallel edges $E_{ij}$ and an endpoint $i$, call a set $P\subseteq E_{ij}$ an $i$-\emph{majority set} if $v_i(P) \ge \frac{1}{2} v_i(E_{ij})$. 
It is said to be \emph{inclusion-minimal} if $v_i(P\setminus\{e\})<\frac{1}{2} v_i(E_{ij})$ for every $e\in P$.

\begin{algorithm}[t]
\caption{\textsc{Compute-PROPm}}
\label{alg:main}
\begin{algorithmic}[1]
\Require A multigraph goods instance $I=(G,\mathbf v)$.
\Ensure A PROPm orientation $\pi$.
\State Initialize $\pi_i\gets\emptyset$ for every agent $i$.
\State $(G',\pi,\mathcal P)\gets\textsc{Pairwise-Resolve}(G,\mathbf v,\pi)$. \Comment{Subroutine~\ref{sub:pairwise}}
\State $\mathcal G\gets\textsc{Build-Virtual-Graph}(G',\mathbf v,\mathcal P)$. \Comment{Subroutine~\ref{sub:virtual_graph}}
\State $\Pi\gets\textsc{Greedy-Walk}(\mathcal G)$.\label{step::alg::main-greedy} \Comment{Subroutine~\ref{sub:greedy_walk}}
\For{each virtual edge $\{u_i,u_j\}$ corresponding to the unresolved pair $E_{ij}$}\label{step::alg::main-final-start}
    \If{$\{u_i,u_j\}$ is oriented toward $u_i$ in $\Pi$}
        \State $\pi_i\gets \pi_i\cup P_{ij}$ and
        $\pi_j\gets \pi_j\cup(E_{ij}\setminus P_{ij})$.
    \Else\, $\pi_j\gets \pi_j\cup P_{ij}$ and
        $\pi_i\gets \pi_i\cup(E_{ij}\setminus P_{ij})$.
    \EndIf
\EndFor\label{step::alg::main-final-end}
\State \Return $\pi$.
\end{algorithmic}
\end{algorithm}

\begin{subroutine}[t]
\caption{\textsc{Min-Majority-Partition}}
\label{sub:min_majority}
\begin{algorithmic}[1]
\Require A parallel edge set $E_{ij}$ and a valuation $v_i$.
\Ensure An inclusion-minimal $i$-majority set $P\subseteq E_{ij}$.
\State Let $e_1,\ldots,e_{|E_{ij}|}$ be an arbitrary ordering of $E_{ij}$.
\State $P\gets E_{ij}$.
\For{$k=1$ to $|E_{ij}|$}
    \If{$v_i(P\setminus\{e_k\})\ge \frac{1}{2} v_i(E_{ij})$}
        \State $P\gets P\setminus\{e_k\}$.
    \EndIf
\EndFor
\State \Return $P$.
\end{algorithmic}
\end{subroutine}

\begin{restatable}{lemma}{lemPartition}\label{lem:partition}
For every adjacent pair $\{i,j\}$, \textsc{Min-Majority-Partition} (Subroutine~\ref{sub:min_majority}) returns a
inclusion-minimal $i$-majority set $P_{ij}$. Moreover, exactly one of the
following two cases holds:
\begin{enumerate}[leftmargin=2.4em]
    \item[(C1)]
    $v_i(P_{ij})\ge \frac{1}{2} v_i(E_{ij})$ and
    $v_j(P_{ij})<\frac{1}{2} v_j(E_{ij})$.
    \item[(C2)]
    $v_i(P_{ij})\ge \frac{1}{2} v_i(E_{ij})$ and
    $v_j(P_{ij})\ge \frac{1}{2} v_j(E_{ij})$.
\end{enumerate}
\end{restatable}

The proof of \Cref{lem:partition} is deferred to Appendix~\ref{sec:lem-partitionAppendix}.

\begin{subroutine}[t]
\caption{\textsc{Pairwise-Resolve}}
\label{sub:pairwise}
\begin{algorithmic}[1]
\Require A multigraph $G$, valuation profile $\mathbf v$, and a partial orientation $\pi$.
\Ensure A reduced graph $G'$, an updated partial orientation $\pi$, and stored sets $\mathcal P$.
\State $G'\gets G$ and $\mathcal P\gets\emptyset$.
\For{each (unordered) adjacent pair $\{i,j\}$ in $G$}
    \State $P_{ij}\gets\textsc{Min-Majority-Partition}(E_{ij},v_i)$.
    \If{$P_{ij}=\emptyset$} \Comment{This happens when $v_i(E_{ij})=0$}
        \State $\pi_i\gets \pi_i\cup P_{ij}$ and
        $\pi_j\gets \pi_j\cup(E_{ij}\setminus P_{ij})$.
        \State Remove $E_{ij}$ from $G'$.
    \ElsIf{(C1) holds}
        \State $\pi_i\gets \pi_i\cup P_{ij}$ and
        $\pi_j\gets \pi_j\cup(E_{ij}\setminus P_{ij})$.
        \State Remove $E_{ij}$ from $G'$.
    \ElsIf{(C2) holds and there exists $e\in P_{ij}$ such that
    $v_j(e)\le v_j(P_{ij})-\frac{1}{2} v_j(E_{ij})$}
        \State Choose such an edge $e$.
        \State $\pi_i\gets
        \pi_i\cup\big(E_{ij}\setminus(P_{ij}\setminus\{e\})\big)$.
        \State $\pi_j\gets \pi_j\cup(P_{ij}\setminus\{e\})$.
        \State Remove $E_{ij}$ from $G'$.
    \Else \, store $P_{ij}$ in $\mathcal P$.
    \EndIf
\EndFor
\State Remove isolated vertices from $G'$.
\State \Return $(G',\pi,\mathcal P)$.
\end{algorithmic}
\end{subroutine}

\begin{restatable}{lemma}{lemPairWise}\label{lem:pair-wise}
Every pair removed by \textsc{Pairwise-Resolve} (Subroutine~\ref{sub:pairwise}) is oriented so that both
endpoints receive at least half value for the edges between them.
\end{restatable}

\begin{proof}
If $P_{ij}=\emptyset$, then $v_i(E_{ij})=0$, because $P_{ij}$ is an
$i$-majority set. Assigning $P_{ij}$ to $i$ and $E_{ij}\setminus P_{ij}$ to
$j$ gives agent $i$ value $\frac{1}{2}v_i(E_{ij})=0$ and gives agent $j$ her full
value for $E_{ij}$.
If (C1) holds, then assigning $P_{ij}$ to $i$ gives agent $i$ value at least
$\frac{1}{2}v_i(E_{ij})$. Since $v_j(P_{ij})<\frac{1}{2}v_j(E_{ij})$, assigning
$E_{ij}\setminus P_{ij}$ to $j$ gives agent $j$ value strictly greater than
$\frac{1}{2}v_j(E_{ij})$.
Suppose (C2) holds and there exists $e\in P_{ij}$ such that
$v_j(e)\le v_j(P_{ij})-\frac{1}{2}v_j(E_{ij})$.
Then $v_j(P_{ij}\setminus\{e\})\ge \frac{1}{2}v_j(E_{ij})$.
By the minimality of $P_{ij}$ for agent $i$, we have $v_i(P_{ij}\setminus\{e\})<\frac{1}{2}v_i(E_{ij})$.
Hence assigning $E_{ij}\setminus(P_{ij}\setminus\{e\})$ to $i$ gives agent
$i$ value strictly greater than $\frac{1}{2}v_i(E_{ij})$, while assigning
$P_{ij}\setminus\{e\}$ to $j$ gives agent $j$ value at least
$\frac{1}{2}v_j(E_{ij})$.
\end{proof}

\begin{subroutine}[t]
\caption{\textsc{Build-Virtual-Graph}}
\label{sub:virtual_graph}
\begin{algorithmic}[1]
\Require Reduced graph $G'$, valuation profile $\mathbf v$, and stored sets $\mathcal P$.
\Ensure A simple virtual graph $\mathcal G=(\mathcal N,\mathcal E)$ with endpoint weights.
\State Create one virtual vertex $u_i$ for every agent $i$ in $G'$.
\State Initialize $\mathcal E\gets\emptyset$.
\For{each adjacent pair $\{i,j\}$ in $G'$}
    \State Add the virtual edge $\{u_i,u_j\}$ to $\mathcal E$.
    \State $w_i(\{u_i,u_j\})\gets
    2\left(v_i(P_{ij})-\frac{1}{2} v_i(E_{ij})\right)$.
    \State $w_j(\{u_i,u_j\})\gets
    2\left(v_j(P_{ij})-\frac{1}{2} v_j(E_{ij})\right)$.
\EndFor
\State \Return $\mathcal G$.
\end{algorithmic}
\end{subroutine}

\begin{subroutine}[t]
\caption{\textsc{Greedy-Walk}}
\label{sub:greedy_walk}
\begin{algorithmic}[1]
\Require A simple virtual graph $\mathcal G=(\mathcal N,\mathcal E)$ with nonnegative endpoint weights.
\Ensure An orientation $\Pi$ of $\mathcal E$.
\State $\Pi\gets\emptyset$ and $\mathcal E_{\mathrm{rem}}\gets\mathcal E$.
\If{$\mathcal E_{\mathrm{rem}}=\emptyset$} \State \Return $\Pi$. \EndIf
\State Let $u$ be an arbitrary vertex incident to an edge in $\mathcal E_{\mathrm{rem}}$.
\While{$\mathcal E_{\mathrm{rem}}\neq\emptyset$}
    \If{$u$ has no incident edge in $\mathcal E_{\mathrm{rem}}$}\Comment{Handle isolated vertex}
        \State Let $u$ be an arbitrary vertex incident to an edge in $\mathcal E_{\mathrm{rem}}$.
    \EndIf
    \State Let $\{u,v\}\in
    \arg\max_{\{u,z\}\in\mathcal E_{\mathrm{rem}}}w_u(\{u,z\})$.
    \State Orient $\{u,v\}$ toward $u$.
    \State Remove $\{u,v\}$ from $\mathcal E_{\mathrm{rem}}$.
    \State $u\gets v$. \Comment{Transfer the turn to the neighbor}
\EndWhile
\State \Return $\Pi$.
\end{algorithmic}
\end{subroutine}

For an orientation $\Pi$ of $\cG$, let $W_i(\Pi)$ be the set of virtual edges
incident to $u_i$ that are oriented toward $u_i$, and let $L_i(\Pi)$ be the
set of virtual edges incident to $u_i$ that are oriented away from $u_i$.
For simplicity, we call each edge in $W_i(\Pi)$ a \emph{won} edge and call each edge in $L_i(\Pi)$ a \emph{lost} edge.

\begin{restatable}{lemma}{lemGreedyWalk} \label{lem:greedy-walk}
For every virtual vertex $u_i$, $\Pi$ returned by \textsc{Greedy-Walk} (Subroutine~\ref{sub:greedy_walk})
satisfies
\[
    \sum_{e\in L_i(\Pi)} w_i(e)
    -
    \sum_{e\in W_i(\Pi)} w_i(e)
    \le
    \max_{e\in L_i(\Pi)} w_i(e),
\]
where the maximum is defined as $0$ if $L_i(\Pi)=\emptyset$.
\end{restatable}

We analyze the subsequence of edges incident to $u_i$ processed during \textsc{Greedy-Walk} (Subroutine~\ref{sub:greedy_walk}). Because the walk traverses an edge immediately after orienting it, $u_i$ cannot lose two consecutive incident edges. Hence, every lost edge, except possibly the first one, is immediately preceded by a won edge. Since the algorithm greedily selects the heaviest available incident edge, each such won edge has a weight at least that of the subsequent lost edge. The detailed proof is deferred to Appendix~\ref{sec:lem-greedy-walkAppendix}.

\begin{restatable}{lemma}{lemReduced}\label{lem:reduced}
Let $\pi'$ be the orientation of the reduced graph $G'$ obtained by mapping
back the virtual orientation returned by \textsc{Greedy-Walk} (Subroutine~\ref{sub:greedy_walk}). For every agent
$i$, define $\Delta_i'=\frac{1}{2}v_i(E_i')-v_i(\pi_i')$. Then, either
\[
    \Delta_i'\le 0, \quad \text{or}  \quad \Delta_i'
    \le
    \max_{j\in N^{\pi'}_+(i)}
    \min_{e\in E_{ij}\cap \pi'_j}v_i(e).
\]
\end{restatable}

\begin{proof}
Consider an unresolved pair $\{i,j\}$ with stored set $P_{ij}$. Since this
pair was not removed during \textsc{Pairwise-Resolve} (Subroutine~\ref{sub:pairwise}), condition (C2) holds and no
edge $e\in P_{ij}$ satisfies
$
    v_j(e)\le v_j(P_{ij})-\frac{1}{2}v_j(E_{ij})
$.
Hence, for every $e\in P_{ij}$, it holds that
$
    v_j(P_{ij})-\frac{1}{2}v_j(E_{ij})<v_j(e).
$
On the other hand, by minimality of $P_{ij}$ for agent $i$, for every
$e\in P_{ij}$,
$
    v_i(P_{ij})-\frac{1}{2}v_i(E_{ij})<v_i(e).
$
Together with (C2), this implies that for any $\ell\in\{i,j\}$,
\[
    0\le
    2\left(v_\ell(P_{ij})-\frac{1}{2}v_\ell(E_{ij})\right)
    \le
    2\min_{e\in P_{ij}}v_\ell(e).
\]
Thus, the virtual endpoint weights are nonnegative and satisfy
$
    w_\ell(\{u_i,u_j\})
    \le
    2\min_{e\in P_{ij}}v_\ell(e)
$.
Fix an agent $i$. If a virtual edge $e=\{u_i,u_j\}$ is oriented to $u_i$, then
agent $i$ receives $P_{ij}$ on the corresponding pair, contributing
$
    \frac{1}{2}v_i(E_{ij})-v_i(P_{ij})
    =
    -\frac{1}{2} w_i(e)
$
to $\Delta_i'$. If the virtual edge is oriented to $u_j$, then agent $i$ receives
$E_{ij}\setminus P_{ij}$, contributing
$
    \frac{1}{2}v_i(E_{ij})-v_i(E_{ij}\setminus P_{ij})
    =
    \frac{1}{2} w_i(e)
$
to $\Delta_i'$. By Lemma~\ref{lem:greedy-walk},
\[
\begin{aligned}
    \Delta_i'
    =
    \frac{1}{2}\sum_{e\in L_i(\Pi)}w_i(e)
    -
    \frac{1}{2}\sum_{e\in W_i(\Pi)}w_i(e)
    \le
    \frac{1}{2}\max_{e\in L_i(\Pi)}w_i(e),
\end{aligned}
\]

If $L_i(\Pi)=\emptyset$, then $\Delta_i'\le 0$. Otherwise, let
$e^*=\{u_i,u_j\}\in L_i(\Pi)$ maximize $w_i(e)$ over $L_i(\Pi)$. Since
$i$ loses this virtual edge, the algorithm assigns $P_{ij}$ to $j$. Hence
$
    E_{ij}\cap \pi'_j=P_{ij}.
$
Using the bound on virtual weights, we have
\[
    \Delta_i'
    \le
    \frac{1}{2} w_i(e^*)
    \le
    \min_{e\in P_{ij}}v_i(e)
    =
    \min_{e\in E_{ij}\cap \pi'_j}v_i(e).
\]
Taking the maximum over all neighbors that receive edges from $i$ gives the
claimed inequality.
\end{proof}

\begin{proof}[Proof of Theorem~\ref{thm:propmExists}]
Run Algorithm~\ref{alg:main} and let $\pi$ be the returned orientation. Fix an agent $i$. Let $E_i'$ be the set of edges
incident to $i$ in the reduced graph $G'$, and let $\pi_i'$ be the part of
$i$'s final bundle that comes from $E_i'$.
By Lemma~\ref{lem:pair-wise}, for every pair $\{i,j\}$ with $E_{i,j}$ removed in \textsc{Pairwise-Resolve}, the edges in $E_{i,j}$ are allocated so that both endpoints receive at least half of their value for $E_{i,j}$.
Therefore, the removed edges can only decrease
$\frac{1}{2}v_i(E_i)-v_i(\pi_i)$ relative to the same quantity restricted to the
reduced graph. Hence,
$
    \frac{1}{2}v_i(E_i)-v_i(\pi_i)
    \le
    \frac{1}{2}v_i(E_i')-v_i(\pi_i').
$

Let $\Delta_i'=\frac{1}{2}v_i(E_i')-v_i(\pi_i')$. If $\Delta_i'\le 0$, then $v_i(\pi_i)\ge \frac{1}{2}v_i(E_i)$, so agent $i$ receives her proportional share and satisfies PROPm.
Now suppose $\Delta_i'>0$. By Lemma~\ref{lem:reduced},
$
    \Delta_i'
    \le
    \max_{j\in N^{\pi'}_+(i)}
    \min_{e\in E_{ij}\cap \pi'_j}v_i(e)
$.
For every unresolved pair $\{i,j\}$, the allocation of $E_{ij}$ under $\pi$ is the same as that under $\pi'$. Hence $N^{\pi'}_+(i)\subseteq N^{\pi}_+(i)$, and for every
$j\in N^{\pi'}_+(i)$, we have
$E_{ij}\cap \pi'_j=E_{ij}\cap \pi_j$. Therefore, $\Delta_i'$ is at most the
PROPm correction of agent $i$ under the final orientation $\pi$.
Combining this with
$\frac{1}{2}v_i(E_i)-v_i(\pi_i)\le \Delta_i'$, we obtain
\[
    v_i(\pi_i)
    +
    \max_{j\in N^{\pi}_+(i)}
    \min_{e\in E_{ij}\cap \pi_j}v_i(e)
    \ge
    \frac{1}{2}v_i(E_i)
    =
    \mathsf{PROP}_i.
\]
Thus, agent $i$ satisfies PROPm. The returned orientation is PROPm for any $i$.

Algorithm~\ref{alg:main} always terminates in
$O(m\log m)$ time. For each adjacent pair $\{i,j\}$, \textsc{Min-Majority-Partition} (Subroutine~\ref{sub:min_majority}) runs in
$O(|E_{ij}|)$ time, and the checks in \textsc{Pairwise-Resolve} (Subroutine~\ref{sub:pairwise}) require $O(|E_{ij}|)$ time. Summing over all adjacent pairs yields a pairwise-resolution phase that takes $O(m)$ time.
The virtual graph has at most $m$ edges. Building it and converting the virtual orientation back to the original multigraph both take $O(m)$ time.
For \textsc{Greedy-Walk} (Subroutine~\ref{sub:greedy_walk}), sort the incident virtual edges of every virtual vertex in decreasing order of the corresponding endpoint weight. The total sorting time is $O(m\log m)$. After sorting, each virtual edge is processed and skipped only a constant number of times, so the greedy walk takes $O(m)$ time. The total running time is therefore
$O(m\log m)$.
\end{proof}

\section{The Interaction between PROPm and Efficiency}
In this section, we study the interaction between fairness and efficiency. We first quantify the social welfare loss under the fairness constraint by establishing the tight bound for the price of fairness.
We then present the computational complexity of computing the PROPm orientation with maximum social welfare.
Finally, we study the compatibility between PROPm and PO.

\subsection{The Price of PROPm} \label{sec:priceOfPROPm}

To consider the price of PROPm, we first notice that the PROPm orientation returned by Algorithm~\ref{alg:main} may not have a bounded social welfare guarantee. 
In this section, we design a more involved algorithm that ensures social welfare at least $\frac{1}{2} \sum_{e=\{i,j\}\in E}\big(v_i(e)+v_j(e)\big)$.


\begin{restatable}{theorem}{thmPropmPrice} \label{thm:propmPrice}
    The price of PROPm orientations for goods is 2, which is tight. Moreover, we can compute a PROPm orientation achieving at least half of the optimal welfare in polynomial time.
\end{restatable}

We begin with a constructive proof of the upper bound. 

\begin{restatable}{proposition}{propPriceUpperBound}\label{prop:priceUpperBound}
For allocations of goods, there exists a PROPm orientation achieving at least half of the optimal welfare. Moreover, such a PROPm orientation can be computed in polynomial time.
\end{restatable}

To prove Proposition~\ref{prop:priceUpperBound}, we use the same algorithm as
in the proof of Theorem~\ref{thm:propmExists}, except that we replace
\textsc{Greedy-Walk} (Subroutine~\ref{sub:greedy_walk}) by \textsc{Orient-Virtual-Graph} (Subroutine~\ref{sub:orient_virtual}).
The orientation of $\cG$ computed by \textsc{Orient-Virtual-Graph} has two properties. 
First, it (and its reversal achieved by reversing every edge) satisfies the same
local deficit bound needed for PROPm. Second, by comparing it with its reversal,
we can choose one of the two orientations whose total virtual weight is at
least half of the total virtual endpoint weight.
In the auxiliary graph $H$ constructed in \textsc{Orient-Virtual-Graph} (Subroutine~\ref{sub:orient_virtual}), a vertex $x_e$ represents an edge $e$ of the
virtual graph $\mathcal G$.
An edge of $H$ labeled by $i$ records that the two corresponding virtual edges were paired in the sorted list of edges incident to $u_i$. 
The properties of $H$ are stated in \Cref{prop:auxiliary-graph}, with proof deferred in Appendix~\ref{sec:auxiliary-graphAppendix}.

\begin{subroutine}[t]
\caption{\textsc{Orient-Virtual-Graph}}
\label{sub:orient_virtual}
\begin{algorithmic}[1]
\Require A simple virtual graph $\mathcal G=(\mathcal N,\mathcal E)$ with endpoint weights
\Ensure An orientation $\Gamma$ of $\mathcal E$
\State Create an auxiliary graph $H$ as follows.
\State For every edge $e\in\mathcal E$, create one vertex $x_e$ in $H$.
\For{each virtual vertex $u_i\in\mathcal N$}
    \State Let $e_1,\ldots,e_d$ be the edges incident to $u_i$.
    \State Sort $e_1,\ldots,e_d$ so that $w_i(e_1) \geq w_i(e_2) \geq \cdots \geq w_i(e_d)$.
    \For{$r=1$ to $\lfloor\frac{1}{2}d\rfloor$}
        \State Add an edge between $x_{e_{2r-1}}$ and $x_{e_{2r}}$ in $H$,
        labeled by $i$.
    \EndFor
\EndFor
\State In each path component of $H$, choose one endpoint and orient all edges away from it.
\State In each cycle component of $H$, orient all edges in any cyclic  direction.
\State Convert this orientation of $H$ into an orientation $\Pi$ of $\mathcal G$ as follows.
\For{each edge $e=\{u_i,u_j\}\in\mathcal E$}\label{step::alg::sub5::convert-start}
    \If{$x_e$ has an incoming edge in $H$ labeled by $i$}
        \State Orient $e$ toward $u_i$ in $\Pi$.
    \ElsIf{$x_e$ has an incoming edge in $H$ labeled by $j$}
        \State Orient $e$ toward $u_j$ in $\Pi$.
    \Else
        \State Orient $e$ toward an endpoint $u_k\in\{u_i,u_j\}$ such that
        $x_e$ is not incident to any edge of $H$ labeled by $k$.\label{step::alg::sub5::convert-ends}
        \Comment{If both endpoints satisfy this, choose one arbitrarily}
    \EndIf
\EndFor
\State Let $\widehat{\Pi}$ be the orientation obtained from $\Pi$ by reversing every edge of $\mathcal G$.
\State For each $u_i\in\mathcal N$, let $W_i(\Pi)$ and $W_i(\widehat{\Pi})$ be the
sets of edges incident to $u_i$ that are oriented toward $u_i$ under $\Pi$ and
$\widehat{\Pi}$, respectively.
\If{$\sum_{u_i\in\mathcal N}\sum_{e\in W_i(\Pi)}w_i(e)\ge
\sum_{u_i\in\mathcal N}\sum_{e\in W_i(\widehat{\Pi})}w_i(e)$}
    \State \Return $\Pi$
\Else\, \Return $\widehat{\Pi}$
\EndIf
\end{algorithmic}
\end{subroutine}

\begin{restatable}{proposition}{propAuxiliaryGraph} \label{prop:auxiliary-graph}
The auxiliary graph $H$ is simple. Every edge of $H$ has exactly one label,
and every vertex of $H$ has degree at most $2$. After the path and cycle
orientations in \textsc{Orient-Virtual-Graph} (Subroutine~\ref{sub:orient_virtual}), every vertex $x_e$ has at
most one incoming edge.
\end{restatable}

The above property makes the conversion rule in Steps~\ref{step::alg::sub5::convert-start}--\ref{step::alg::sub5::convert-ends} of \textsc{Orient-Virtual-Graph} (Subroutine~\ref{sub:orient_virtual}) well-defined. 
If $x_e$ has an incoming edge, its unique label determines the endpoint toward which $e$ is
oriented. If $x_e$ has no incoming edge, then $x_e$ is either isolated or
the first vertex of a directed path. Hence at least one endpoint of $e$ did
not pair $e$ in the construction of $H$, and the algorithm orients $e$
toward such an endpoint.

\begin{restatable}{lemma}{lemPricePropm}\label{lem:price-propm}
Let $\Pi$ be the virtual orientation produced from the oriented auxiliary graph
$H$, and let $\widehat{\Pi}$ be the orientation obtained by reversing every
edge of $\mathcal G$. If either $\Pi$ or $\widehat{\Pi}$ is used in Steps~\ref{step::alg::main-final-start}--\ref{step::alg::main-final-end} of Algorithm~\ref{alg:main}, the resulting orientation is PROPm.
\end{restatable}

The key to \Cref{lem:price-propm} is to show that both $\Pi$ and $\widehat{\Pi}$ satisfy the same local deficit bound as in Lemma~\ref{lem:greedy-walk}. By construction, the auxiliary graph $H$ pairs the incident virtual edges at each vertex in decreasing order of their weights. The orientation rules for the paths and cycles of $H$ guarantee a ``one-won-one-lost'' property for each pair under both $\Pi$ and its reversal $\widehat{\Pi}$. Because the pairs are sorted by weight, any lost edge can be injectively charged to a won edge from the preceding pair, which has at least the same weight. Consequently, the total weight of lost edges exceeds the total weight of won edges by at most the weight of a single lost edge. 
By applying arguments similar to the proof of Theorem~\ref{thm:propmExists}, we can prove the PROPm condition. 
The detailed proof is deferred to Appendix~\ref{sec:lem-price-propmAppendix}. It remains to choose between the two PROPm orientations. The following proof of \Cref{prop:priceUpperBound} shows
that one of them has welfare at least half of the unconstrained optimum
welfare.
\enlargethispage{\baselineskip}

\begin{proof}[Proof of Proposition~\ref{prop:priceUpperBound}]
Run Algorithm~\ref{alg:main}, replacing \textsc{Greedy-Walk} (Subroutine~\ref{sub:greedy_walk}) by
\textsc{Orient-Virtual-Graph} (Subroutine~\ref{sub:orient_virtual}). By Lemma~\ref{lem:price-propm}, the returned
orientation $\pi$ is PROPm.

Let $R$ be the set of original edges contained in edge sets resolved during the pairwise-resolution step, and let $U=E\setminus R$ be the remaining edges.
By Lemma~\ref{lem:pair-wise}, every resolved edge set $E_{ij}$ gives both endpoints at least half of their value for $E_{ij}$. 
Thus, the welfare from the edges in $R$ is at least $\frac{1}{2} \sum_{e=\{i,j\}\in R}\bigl(v_i(e)+v_j(e)\bigr)$.
Now consider the unresolved edge sets.
Let $\Gamma$ be the virtual orientation returned by \textsc{Orient-Virtual-Graph} (Subroutine~\ref{sub:orient_virtual}). 
For each virtual vertex $u_i$, let $W_i(\Gamma)$ be the set of edges incident to $u_i$ that are oriented toward $u_i$, and let $L_i(\Gamma)$ be the set of edges incident to $u_i$ that are oriented away from $u_i$. Since \textsc{Orient-Virtual-Graph} (Subroutine~\ref{sub:orient_virtual}) returns the better of $\Pi$ and $\widehat{\Pi}$,
\[
    \sum_{u_i\in\mathcal N}\sum_{e\in W_i(\Gamma)}w_i(e)
    \ge
    \frac{1}{2}
    \sum_{u_i\in\mathcal N}
    \sum_{e'\text{ incident to }u_i}w_i(e').
\]
Equivalently, the total virtual weight of edges won under $\Gamma$ is at least the total virtual weight of edges lost under $\Gamma$.
For an unresolved edge set $E_{ij}$ represented by the virtual edge $e'=\{u_i,u_j\}$, suppose first that $e'$ is oriented toward $u_i$ under $\Gamma$. 
Then agent $i$ receives $P_{ij}$ and agent $j$ receives $E_{ij}\setminus P_{ij}$. By the definition of the endpoint weights,
$v_i(P_{ij}) = \frac{1}{2}v_i(E_{ij})+\frac{1}{2}w_i(e')$, $v_j(E_{ij}\setminus P_{ij}) = \frac{1}{2}v_j(E_{ij})-\frac{1}{2}w_j(e')$.
The analogous identity holds when $e'$ is oriented toward $u_j$. 
Summing over all unresolved edge sets, the welfare \enlargethispage{\baselineskip} from the edges in $U$ is
\[
    \frac{1}{2}
    \sum_{e=\{i,j\}\in U}\bigl(v_i(e)+v_j(e)\bigr)
    +
    \frac{1}{2}
    \left(
        \sum_{u_i\in\mathcal N}\sum_{e\in W_i(\Gamma)}w_i(e)
        -
        \sum_{u_i\in\mathcal N}\sum_{e'\in L_i(\Gamma)}w_i(e')
    \right).
\]
The second term is nonnegative by the choice of $\Gamma$. Hence the welfare
from $U$ is at least
$
    \frac{1}{2}
    \sum_{e=\{i,j\}\in U}\bigl(v_i(e)+v_j(e)\bigr)
$.
Combining the resolved and unresolved edges gives
$$
    \mathsf{SW}(\pi)
    \ge
    \frac{1}{2}
    \sum_{e=\{i,j\}\in E}\bigl(v_i(e)+v_j(e)\bigr).
$$
For any orientation, the welfare contribution of an edge $e=\{i,j\}$ is at
most $\max\{v_i(e),v_j(e)\}$, which is at most $v_i(e)+v_j(e)$. Thus, the optimal welfare is at most
$
    \sum_{e=\{i,j\}\in E}\bigl(v_i(e)+v_j(e)\bigr)
    \le
    2\mathsf{SW}(\pi)
$, which implies that the price of PROPm is at most 2.

For the running time, in \textsc{Orient-Virtual-Graph}, sorting the incident virtual edges at all
vertices takes $O(m\log m)$ time in total. Constructing the auxiliary
graph $H$, orienting its path and cycle components, converting the resulting
orientation back to an orientation of $\mathcal G$, and comparing $\Pi$ with
$\widehat{\Pi}$ all take linear time. Hence the modified algorithm also runs in
$O(m\log m)$ time.
\end{proof}

After that, we are going to present the lower bound through \Cref{prop:priceLowerBound} with the proof deferred to Appendix~\ref{sec:priceLowerBoundAppendix}.

\begin{restatable}{proposition}{propPriceLowerBound}\label{prop:priceLowerBound}
For goods instances, the price of PROPm is at least $2$.
\end{restatable}

As a result, Theorem~\ref{thm:propmPrice} follows directly from Propositions~\ref{prop:priceUpperBound} and \ref{prop:priceLowerBound}.

\subsection{The Complexity of the Optimization Problem}
\label{sec:optimization:welfare}

In this section, we show that computing a PROPm
orientation with maximum social welfare is APX-hard on multigraphs and NP-hard on simple graphs, for example a star.
We also present an FPTAS when the graph is a star.
The proofs can be found in Appendix \ref{app:proofs:optimization}. 
It is worth noting that due to the graphical constraints, the existing hardness results in the unconstrained model do not imply the hardness results in our context \cite{DBLP:journals/eor/AzizHMS23,DBLP:conf/sigecom/Bu0LST25}.

\begin{restatable}{theorem}{thmOptWelfareMultiple}\label{thm:optimization:welfare:multiple}
For additive goods valuations, computing a PROPm orientation with maximum social welfare is APX-hard.
\end{restatable}

\vspace{-0.8em} 

\begin{restatable}{theorem}{thmOptWelfareSimple}\label{thm:optimization:welfare:simple}
For additive goods valuations, it is NP-hard to compute a PROPm orientation
with maximum social welfare, even when the graph is simple.
\end{restatable}


\subsection{Compatibility of PROPm with Pareto Optimality}
\label{sec:PROPm:PO}

In this section, we show that PROPm is not generally compatible with PO in the graphical setting and
that deciding whether there exists an orientation that is simultaneously PROPm and PO is
NP-hard. 
Proofs are deferred to Appendix \ref{app:PROPm:PO}.

\begin{restatable}{proposition}{propPropmPoIncompatibleTwoAgents}\label{prop:propmPoIncompatibleTwoAgents}
There exists an instance where no PROPm orientation is PO.
\end{restatable}

\begin{restatable}{theorem}{thmHardPOPROPm}\label{thm:hard:PO:PROPm}
For additive goods valuations, it is NP-hard to decide whether there exists
an orientation that is both PROPm and Pareto optimal, even when all values
belong to $\{0,1,2,4\}$.
\end{restatable}

It is not clear whether the above decision problem belongs to NP. 
While PROPm can be verified in polynomial time for a given orientation, verifying Pareto optimality is coNP-complete \cite{10.1007/978-3-642-04428-1_9}. 
Indeed, it remains weakly coNP-complete already for two agents \cite{aziz2016optimal}. 
This two-agent hardness carries over to the orientation model, which represents each item as a parallel edge between the two agents.

\section{PROPavg Orientations for Goods} \label{sec:PROPavg}

We also consider the notion of PROPavg. Unfortunately, PROPavg orientations may not exist, and even worse, deciding whether they exist is NP-complete.

\begin{restatable}{proposition}{propPropAvgNotExist}\label{prop:prop-avgNotExist}
PROPavg orientations may fail to exist for goods, even in simple graphs with
binary valuations.
\end{restatable}

\begin{figure}[H]
\centering
\begin{tikzpicture}[
    every node/.style={circle, draw, inner sep=1.2pt, minimum size=14pt, font=\scriptsize},
    zeroedge/.style={draw=gray!45, line width=0.35pt},
    spoke/.style={draw=black, line width=0.55pt},
    bridge/.style={draw=black, line width=0.75pt},
    panel/.style={font=\small}
]

\node (a1) at (0,0) {$1$};
\node (a2) at (-1.55,1.15) {$2$};
\node (a3) at (-0.55,1.85) {$3$};
\node (a4) at (0.75,1.75) {$4$};
\node (a5) at (1.55,0.85) {$5$};
\node (a6) at (1.25,-0.55) {$6$};

\node (b1) at (4.2,0) {$1'$};
\node (b2) at (2.65,1.15) {$2'$};
\node (b3) at (3.65,1.85) {$3'$};
\node (b4) at (4.95,1.75) {$4'$};
\node (b5) at (5.75,0.85) {$5'$};
\node (b6) at (5.45,-0.55) {$6'$};

\foreach \x/\y in {a2/a3,a2/a4,a2/a5,a2/a6,a3/a4,a3/a5,a3/a6,a4/a5,a4/a6,a5/a6}
    \draw[zeroedge] (\x) -- (\y);
\foreach \x/\y in {b2/b3,b2/b4,b2/b5,b2/b6,b3/b4,b3/b5,b3/b6,b4/b5,b4/b6,b5/b6}
    \draw[zeroedge] (\x) -- (\y);

\foreach \x in {a2,a3,a4,a5,a6}
    \draw[spoke] (a1) -- (\x);
\foreach \x in {b2,b3,b4,b5,b6}
    \draw[spoke] (b1) -- (\x);

\draw[bridge] (a1) -- (b1)
    node[midway, above, draw=none, rectangle, inner sep=1pt, font=\scriptsize] {$1,1$};

\end{tikzpicture}
\vspace{-0.4cm}
\caption{A compact PROPavg counterexample for goods. Each side is a clique on
six vertices. For every spoke edge between the root and a leaf, the root's
value is $0$, and the leaf's value is $1$. The edge $\{1,1'\}$ has value $1$
for both endpoints. All gray edges have value $0$ for both endpoints.}
\label{fig:PROPavg}
\end{figure}
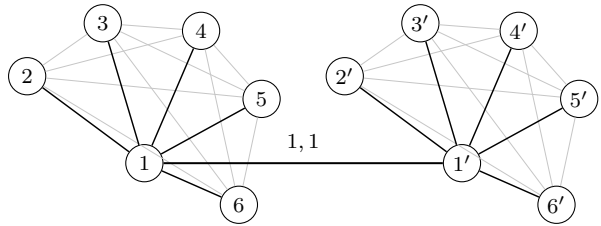

Figure~\ref{fig:PROPavg} illustrates the counterexample for the above proposition. 
The proof of \Cref{prop:prop-avgNotExist} is deferred to Appendix~\ref{sec:prop-avgNotExistAppendix}.

\begin{restatable}{theorem}{thmPropAvgNPC}\label{thm:prop-avgNPC}
Deciding whether a goods instance admits a PROPavg orientation is NP-complete, even for simple graphs with values in $\{0,1,3\}$.
\end{restatable}
\begin{proof}
The problem is in NP, since for a given orientation, one can check the
PROPavg condition for every agent in polynomial time. We prove NP-completeness by
a reduction from \textsc{2P2N-3SAT}~\cite{berman2005approximation}. Let the input formula have variables
$x_1, \ldots, x_n$ and clauses $C_1, \ldots, C_m$, where every clause contains
three distinct literals, and every variable appears exactly twice positively
and exactly twice negatively.

For each variable $x_i$, we create two literal
vertices, denoted by $x_i$ and $\bar x_i$, and add the switch edge
$s_i=\{x_i,\bar x_i\}$. For each clause $C_j$, create one clause vertex
$c_j$. If
$
    C_j=(\ell_j^1\vee \ell_j^2\vee \ell_j^3),
$
then for each $t\in[3]$, add one edge
$e_j^t=\{c_j,a(\ell_j^t)\}$, where $a(x_i)=x_i$ and
$a(\bar x_i)=\bar x_i$.
We also create dummy vertices for each literal vertex. For any literal vertex $a$, attach a
dummy gadget $D(a)$ as follows. Create six dummy vertices
$d_a^1,\ldots,d_a^6$. Add edge $\{a,d_a^1\}$, and add all
edges among the six dummy vertices. Thus the dummy vertices induce a copy of
$K_6$. We attach one copy of this gadget to every literal vertex
$x_i$ and $\bar x_i$.
This completes the graph construction. The graph is simple. It has
$14n+m$ vertices and $33n+3m$ edges.
The valuations are defined as follows. For every switch edge $s_i$, set
$v_{x_i}(s_i)=v_{\bar x_i}(s_i)=3$. For every clause edge
$e_j^t=\{c_j,a(\ell_j^t)\}$, set
$v_{c_j}(e_j^t)=v_{a(\ell_j^t)}(e_j^t)=1$.
For each dummy gadget $D(a)$, set
$  v_a(\{a,d_a^1\})=0$ and $
    v_{d_a^1}(\{a,d_a^1\})=1.
$
For every $r\in\{2,\ldots,6\}$, set
$v_{d_a^1}(\{d_a^1,d_a^r\})=0$ and
$v_{d_a^r}(\{d_a^1,d_a^r\})=1$. Every edge between two vertices in
$\{d_a^2,\ldots,d_a^6\}$ has value $0$ for both endpoints. All valuations are
additive, and all values belong to $\{0,1,3\}$.
The construction is illustrated in Figure~\ref{fig:propavg-reduction}.

For the proportional shares, each literal vertex has a total value $3+2=5$ and hence has proportional share $\frac{5}{2}$. 
Each clause vertex has total value $3$ and hence has a proportional share $\frac{3}{2}$. 
In every dummy gadget, each dummy vertex has a total value $1$ and hence has a proportional share $\frac{1}{2}$.

We first present a property of the dummy gadget.
In any PROPavg orientation, the attachment edge $\{a,d_a^1\}$ must be assigned to $d_a^1$. 
Suppose not, then $d_a^1$ loses its only edge of value $1$.
Since all five edges $\{d_a^1,d_a^r\}$ have value $0$ for $d_a^1$, the PROPavg condition for $d_a^1$ implies that $d_a^1$ can lose at most one of these five edges.
Therefore, at least four vertices among $d_a^2,\ldots,d_a^6$ lose their edge
to $d_a^1$. Each such vertex then loses its only edge of value $1$. To satisfy PROPavg, it
can lose at most one of its four remaining edges to the other vertices in
$\{d_a^2,\ldots,d_a^6\}$. Hence, each of these at least four vertices must
receive at least three edges inside the set $\{d_a^2,\ldots,d_a^6\}$. This
requires at least $12$ assigned edges inside a graph with only
$\binom{5}{2}=10$ edges, a contradiction. Thus $\{a,d_a^1\}$ is assigned to
$d_a^1$ in every PROPavg orientation.

\begin{figure}[t]
\centering
\resizebox{\textwidth}{!}{%
\begin{tikzpicture}[
    agent/.style={circle,draw,thick,minimum size=7mm,inner sep=1pt,font=\small},
    smallagent/.style={circle,draw,minimum size=5.5mm,inner sep=0.5pt,font=\scriptsize},
    literalagent/.style={circle,draw,thick,minimum size=10mm,inner sep=1pt,font=\scriptsize},
    gadgetbox/.style={draw,rounded corners=1.5pt,minimum width=1.45cm,
        minimum height=6.5mm,align=center,font=\scriptsize},
    edge/.style={draw,thick},
    zeroedge/.style={draw=gray!55,thin},
    attach/.style={draw,thick,densely dashed},
    lab/.style={font=\scriptsize,fill=white,inner sep=1pt},
    panel/.style={font=\small}
]

\begin{scope}[shift={(0,0)}]
\node[agent] (a) at (0,0) {$a$};

\node[smallagent] (d1) at (1.35,0) {$d^1$};
\node[smallagent] (d2) at (2.25,0.82) {$d^2$};
\node[smallagent] (d3) at (3.20,0.48) {$d^3$};
\node[smallagent] (d4) at (3.20,-0.48) {$d^4$};
\node[smallagent] (d5) at (2.25,-0.82) {$d^5$};
\node[smallagent] (d6) at (2.55,0) {$d^6$};

\foreach \u/\v in {d2/d3,d2/d4,d2/d5,d2/d6,d3/d4,d3/d5,d3/d6,d4/d5,d4/d6,d5/d6} {
    \draw[zeroedge] (\u)--(\v);
}

\foreach \k in {2,3,4,5,6} {
    \draw[edge] (d1)--(d\k);
}

\draw[edge] (a)--(d1)
    node[lab,midway,above] {$(0,1)$};

\node[lab] at (1.58,0.62) {$(0,1)$};

\node[panel] at (1.75,-1.42) {(a) Dummy gadget $D(a)$};
\end{scope}

\begin{scope}[shift={(5.65,0)}]
\node[agent] (x) at (0,0) {$x_i$};
\node[agent] (nx) at (2.35,0) {$\bar x_i$};

\draw[edge] (x)--(nx) node[lab,midway,above] {$3$};

\node[gadgetbox] (Dx) at (0,-1.05) {$D(x_i)$};
\node[gadgetbox] (Dnx) at (2.35,-1.05) {$D(\bar x_i)$};

\draw[edge] (x)--(Dx) node[lab,midway,left] {$(0,1)$};
\draw[edge] (nx)--(Dnx) node[lab,midway,right] {$(0,1)$};

\draw[edge] (x)--++(-0.78,0.52) node[lab,midway,above right] {$1$};
\draw[edge] (x)--++(-0.78,-0.32) node[lab,midway,above left] {$1$};

\draw[edge] (nx)--++(0.78,0.52) node[lab,midway,above left] {$1$};
\draw[edge] (nx)--++(0.78,-0.32) node[lab,midway,above right] {$1$};

\node[panel] at (1.17,-1.75) {(b) Variable gadget for $x_i$};
\end{scope}

\begin{scope}[shift={(11.35,0)}]
\node[literalagent] (l1) at (0,1.18) {$a(\ell_j^1)$};
\node[literalagent] (l2) at (0,0) {$a(\ell_j^2)$};
\node[literalagent] (l3) at (0,-1.18) {$a(\ell_j^3)$};
\node[agent] (c) at (2.85,0) {$c_j$};

\draw[edge] (l1)--(c) node[lab,midway,above right] {$1$};
\draw[edge] (l2)--(c) node[lab,midway,above] {$1$};
\draw[edge] (l3)--(c) node[lab,midway,above left] {$1$};

\node[panel] at (1.05,-2) {(c) Clause gadget for $C_j$};
\end{scope}

\end{tikzpicture}%
}
\caption{Illustration of the gadgets for the reduction. We suppress
the subscript $a$ and write $d^t$ for $d_a^t$. A single label means that
both endpoints have that value. A label $(0,1)$ means value $0$ at
the left or upper endpoint and value $1$ at the right or lower endpoint. In
(a), $d^1,\ldots,d^6$ induce a clique; the
gray edges have value $0$ for both endpoints. The four dangling edges incident
to $x_i$ and $\bar x_i$ are clause edges.}
\label{fig:propavg-reduction}
\end{figure}
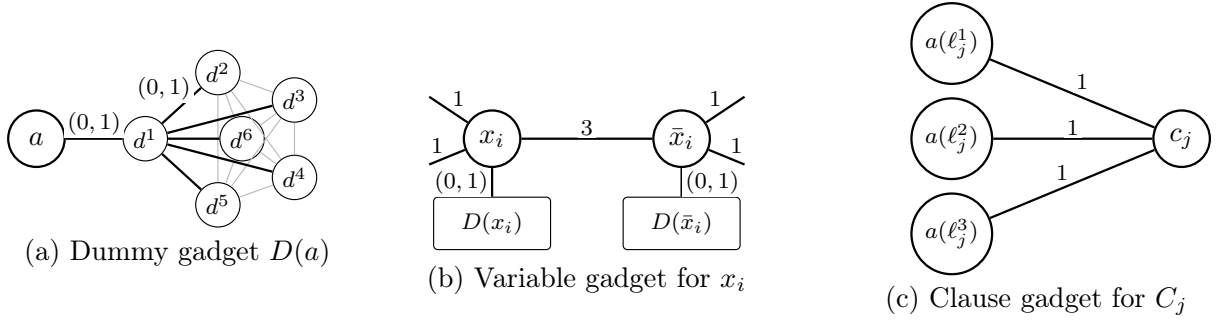
Now suppose the \textsc{2P2N-3SAT} instance has a satisfying assignment. 
We construct an orientation as follows.
If $x_i$ is true, assign $s_i$ to $x_i$; otherwise, assign $s_i$ to $\bar x_i$. 
Assign every edge $\{a,d_a^1\}$ to $d_a^1$ and assign every edge $\{d_a^1,d_a^r\}$ to $d_a^r$. 
The remaining dummy edges have value $0$ for both endpoints and may
be oriented arbitrarily. Finally, assign a clause edge $e_j^t$ to $c_j$ if the literal $\ell_j^t$ is true and assign it to $a(\ell_j^t)$ otherwise.
Every dummy vertex is PROPavg by construction. 
A true literal vertex receives its switch edge and has value $3> \frac{5}{2}$. 
A false literal vertex receives its two clause edges, of total value $2$ and loses the switch edge of value $3$ and
an edge (the one between it and the dummy vertex) of value $0$. 
Thus its PROPavg correction is at least $\frac{3+0}{2}$, so it satisfies PROPavg. Finally, every clause has at least one true literal. Hence every clause vertex receives at least one clause edge and satisfies PROPavg.
Conversely, suppose that the constructed instance admits a PROPavg orientation.
By the dummy-gadget property above, every edge between the literal vertex and the dummy vertex is assigned to the
dummy vertex. 
Define a truth assignment by setting $x_i$ to true if and only if $s_i$ is assigned to $x_i$.

We claim that any literal vertex that does not receive its switch edge must receive both of its clause edges. Indeed, such a vertex loses the switch edge, whose value is $3$, and loses the edge incident to the dummy agent, whose value is $0$. 
If it receives at most one of its two clause edges, then its own value is at most $1$, and its PROPavg correction is at most $\frac{4}{3}$. 
This gives total value less than $\frac{5}{2}$, contradicting PROPavg. Now consider any clause $C_j$. If $C_j$ were not satisfied, then none of its three literal vertices would receive its switch edge. 
By the claim, all three clause edges incident to $c_j$ would be assigned to the literal vertices, so $c_j$ would receive value $0$. Its PROPavg correction would be exactly $1$,
which is less than its proportional share $\frac{3}{2}$.
This contradicts PROPavg.
Therefore every clause is satisfied.
The reduction is in polynomial time, the graph is simple, and all values belong to $\{0,1,3\}$. 
Hence, deciding whether a PROPavg orientation exists is NP-complete.
\end{proof}


\section{Conclusion}

 Our results identify a sharp boundary among proportionality relaxations in the graphical allocation model: PROPm orientations always exist and can be computed efficiently for goods, chores, and mixed instances, whereas PROPavg orientations may not.
Thus, unlike in the unconstrained setting, graphical feasibility separates PROPm from the stronger PROPavg requirement.
We also quantified the efficiency cost of imposing PROPm. For goods, the price of PROPm is exactly $2$, and it is unbounded for chores. We also showed that optimizing welfare subject to PROPm is computationally hard and that PROPm is not generally compatible with PO. These findings suggest several directions for future work, including identifying graph classes where PROPavg can be guaranteed, approximating welfare-maximizing PROPm orientations, and extending proportional fairness beyond edge-orientation models.

\section*{Acknowledgments}

Bo Li is supported in part by the Hong Kong SAR Research Grants Council (Grant No. PolyU 15224823) and the Guangdong Basic and Applied Basic Research Foundation (Grant No. 2024A1515011524).

\bibliographystyle{plain}
\IfFileExists{mylib.bib}{%
  \bibliography{mylib}
}{%
  \bibliography{../mylib}
}

\newpage
\appendix
\section*{Appendix}

\section{Results for Chores and Mixed Items}\label{sec:chores-mixed}

We now extend the model from goods to chores and mixed items. In a chores
instance, each agent $i$ has a nonnegative additive cost function $c_i$ over
incident edges. In a mixed instance, an incident edge is either a good or a chore for each
endpoint. If $e$ is a good for agent $i$, then it has value $v_i(e)\ge 0$; if
$e$ is a chore for agent $i$, then it has cost $c_i(e)>0$. We treat value-zero
edges as goods. Hence each incident edge belongs to exactly one of
$E_i^g$ and $E_i^c$.
For a mixed instance, define the utility of agent $i$ for a set
$S\subseteq E_i$ by
\[
    u_i(S)=\sum_{e\in S\cap E_i^g}v_i(e)-\sum_{e\in S\cap E_i^c}c_i(e).
\]
The proportional share of agent $i$ is
$\mathsf{PROP}_i=\frac{1}{2} u_i(E_i)$.

For an orientation $\pi$, define
\[
    E_{ij}^{g,-}(\pi)=E_{ij}\cap \pi_j\cap E_i^g
    \quad\text{and}\quad
    E_{ij}^{c,+}(\pi)=E_{ij}\cap \pi_i\cap E_i^c .
\]
Thus $E_{ij}^{g,-}(\pi)$ is the set of goods that agent $i$ loses to $j$,
while $E_{ij}^{c,+}(\pi)$ is the set of chores that agent $i$ receives from
$j$. Let
\[
    N_i^{g,-}(\pi)=\{j\in N_i:E_{ij}^{g,-}(\pi)\neq\emptyset\},
    \quad
    N_i^{c,+}(\pi)=\{j\in N_i:E_{ij}^{c,+}(\pi)\neq\emptyset\}.
\]

For convenience, define
\[
    M_i^g(\pi)=
    \max_{j\in N_i^{g,-}(\pi)}
    \min_{e\in E_{ij}^{g,-}(\pi)}v_i(e),
    \quad
    M_i^c(\pi)=
    \max_{j\in N_i^{c,+}(\pi)}
    \min_{e\in E_{ij}^{c,+}(\pi)}c_i(e),
\]
where the maximum is taken to be $0$ if the corresponding neighbor set is
empty.

\begin{definition}[PROPm for mixed items]\label{def:mixed-propm}
An orientation $\pi$ of a mixed instance is \emph{PROPm} if, for every agent
$i\in N$,
\[
    u_i(\pi_i)+\max\{M_i^g(\pi),M_i^c(\pi)\}\ge \mathsf{PROP}_i .
\]
\end{definition}

For PROPavg, define
\[
    A_i^g(\pi)=
    \frac{1}{|N_i^{g,-}(\pi)|}
    \sum_{j\in N_i^{g,-}(\pi)}
    \min_{e\in E_{ij}^{g,-}(\pi)}v_i(e),
\]
when $N_i^{g,-}(\pi)\neq\emptyset$, and set $A_i^g(\pi)=0$ otherwise. Also define
\[
    A_i^c(\pi)=
    \frac{1}{|N_i^{c,+}(\pi)|}
    \sum_{j\in N_i^{c,+}(\pi)}
    \min_{e\in E_{ij}^{c,+}(\pi)}c_i(e),
\]
when $N_i^{c,+}(\pi)\neq\emptyset$, and set $A_i^c(\pi)=0$ otherwise.

\begin{definition}[PROPavg for mixed items]\label{def:mixed-propavg}
An orientation $\pi$ of a mixed instance is \emph{PROPavg} if, for every agent
$i\in N$,
\[
    u_i(\pi_i)+\max\{A_i^g(\pi),A_i^c(\pi)\}\ge \mathsf{PROP}_i .
\]
\end{definition}

When all items are goods, the above
definitions reduce to the goods definitions. When all items are chores,
$u_i(S)=-c_i(S)$, and PROPm becomes
\[
    c_i(\pi_i)-
    \max_{j\in N_-^\pi(i)}
    \min_{e\in E_{ij}\cap \pi_i}c_i(e)
    \le \frac{1}{2} c_i(E_i).
\]
Similarly, PROPavg becomes
\[
    c_i(\pi_i)-
    \frac{1}{|N_-^\pi(i)|}
    \sum_{j\in N_-^\pi(i)}
    \min_{e\in E_{ij}\cap \pi_i}c_i(e)
    \le \frac{1}{2} c_i(E_i),
\]
where the correction is $0$ when $N_-^\pi(i)=\emptyset$.

For chores, the price of fairness is the supremum of the minimum fair social cost divided by the unconstrained minimum social cost, over instances with positive unconstrained minimum cost.
For chores, the social cost of an orientation is
$\mathsf{SC}(\pi)=\sum_{i\in N}c_i(\pi_i)$. Pareto dominance is defined with
respect to costs: an orientation $\pi'$ Pareto dominates $\pi$ if
$c_i(\pi'_i)\le c_i(\pi_i)$ for every agent $i$, and the inequality is strict
for at least one agent.

In what follows, we first present the price of PROPm for chores.

\begin{restatable}{theorem}{thmChoresPriceInfinite}\label{thm:chores-price-infinite}
For chores instances, the price of PROPm is unbounded, even for simple graphs.
\end{restatable}

\begin{proof}
Fix $H>1$. Consider the star with center $0$ and leaves $1$, $2$, and $3$. For each leaf $\ell$, let $e_\ell=\{0,\ell\}$, and set $c_0(e_\ell)=1$ and $c_\ell(e_\ell)=H$.

The unconstrained optimum assigns all three chores to the center, with social cost $3$.
We claim that any PROPm orientation assigns at least one chore to a leaf. 
If all chores are assigned to the center, then the center has cost $3$ and proportional cost share $\frac{3}{2}$. Its PROPm correction can remove only one assigned chore of cost $1$, leaving adjusted cost $2 > \frac{3}{2}$, a contradiction.

Thus every PROPm orientation has social cost at least $H+2$.
This bound is tight: assign one chore to a leaf and the other two to the center. 
The leaf can remove its only assigned chore, the center has adjusted cost $2-1 \leq \frac{3}{2}$, and the other leaves receive no chore. 
Hence the minimum social cost among PROPm orientations is $H+2$.

The price of PROPm on this instance is therefore $\frac{1}{3}(H+2)$, which tends to infinity as $H\to\infty$.
\end{proof}

Next, we establish a goods-chores duality lemma, which is helpful for establishing existence and computation results for chores and mixed items.

\begin{restatable}[Goods--chores duality]{lemma}{lemGoodsChoresDuality}\label{lem:goods-chores-duality}
Let $I=(G,\mathbf c)$ be a chores instance. Construct a goods instance
$I^g=(G,\mathbf v)$ on the same graph by setting $v_i(e)=c_i(e)$ for every
agent $i$ and incident edge $e$. For an orientation $\pi$ of $I^g$, let
$\bar\pi$ be the flipped orientation, where
$\bar\pi_i=E_i\setminus \pi_i$ for every agent $i$. Then:
\begin{enumerate}[leftmargin=2.4em]
    \item $\pi$ is PROPm for $I^g$ if and only if $\bar\pi$ is PROPm for $I$.
    \item $\pi$ is PROPavg for $I^g$ if and only if $\bar\pi$ is PROPavg for $I$.
    \item $\pi$ is Pareto optimal for $I^g$ if and only if $\bar\pi$ is Pareto optimal for $I$.
    \item $\mathsf{SW}_{I^g}(\pi)+\mathsf{SC}_{I}(\bar\pi)=
    \sum_{e=\{i,j\}\in E}\big(c_i(e)+c_j(e)\big)$.
\end{enumerate}
\end{restatable}

\begin{proof}
Fix an agent $i$. Under the flipped orientation, $i$ receives exactly the edges it does not receive under $\pi$. 
Hence $\bar\pi_i=E_i\setminus\pi_i$, $N_+^\pi(i)=N_-^{\bar\pi}(i)$, and
$E_{ij}\cap\pi_j=E_{ij}\cap\bar\pi_i$ for every $j\in N_+^\pi(i)$. Therefore the goods-PROPm condition for $i$ under $\pi$ is
equivalent to
\[
    c_i(\bar\pi_i)-
    \max_{j\in N_-^{\bar\pi}(i)}
    \min_{e\in E_{ij}\cap\bar\pi_i}c_i(e)
    \le \frac{1}{2} c_i(E_i),
\]
which is the chores-PROPm condition for $i$ under $\bar\pi$. The same argument, with the maximum replaced by the average over the same neighbor set, proves the PROPavg statement.

For Pareto optimality, note that
$c_i(\bar\pi_i)=c_i(E_i)-v_i(\pi_i)$ for every agent $i$. Thus increasing an
agent's value in the goods instance is exactly the same as decreasing the same
agent's cost in the flipped chores instance. Pareto improvements are therefore
preserved in both directions.

The welfare-cost identity follows by summing
$c_i(\bar\pi_i)=c_i(E_i)-v_i(\pi_i)$ over all agents.
\end{proof}

We now transfer the goods results to chores. 

\begin{restatable}{theorem}{thmChoresPropmExists}\label{thm:chores-propm-exists}
PROPm orientations for chores are guaranteed to exist in multigraphs. One can
compute a PROPm orientation for chores in $O(m\log m)$ time.
\end{restatable}

\begin{proof}
Given a chores instance $I=(G,\mathbf c)$, construct the goods instance
$I^g=(G,\mathbf v)$ by setting $v_i(e)=c_i(e)$ for every incident edge. By
Theorem~\ref{thm:propmExists}, one can compute a PROPm orientation $\pi$ of
$I^g$ in $O(m\log m)$ time. By Lemma~\ref{lem:goods-chores-duality}, the
flipped orientation $\bar\pi$ is PROPm for the chores instance. Flipping all
edges takes linear time, so the total running time remains $O(m\log m)$.
\end{proof}

\begin{restatable}{theorem}{thmChoresPropmApxHard}\label{thm:chores-propm-apx-hard}
For additive chores costs, computing a PROPm orientation with minimum social
cost is APX-hard in multigraphs.
\end{restatable}

\begin{proof}
Apply Lemma~\ref{lem:goods-chores-duality} to the APX-hardness construction
for goods. In that construction, the total value over all edge-endpoint pairs is
\[
    T=\sum_{e=\{i,j\}\in E}\big(v_i(e)+v_j(e)\big)=\frac{103}{2}m,
\]
and the canonical welfare is $W^\star=\frac{73}{2}m$. Hence
$T-W^\star=15m$.

In the completeness case, there is a PROPm goods orientation with welfare at
least $W^\star-2\epsilon_0m$. Flipping this orientation gives a PROPm chores
orientation with social cost at most $15m+2\epsilon_0m$.

In the soundness case, every PROPm goods orientation has welfare at most $W^\star-(\frac{1}{1016}-\epsilon_0)m$. 
Therefore, every PROPm chores orientation has
social cost at least
\[
    15m+\left(\frac{1}{1016}-\epsilon_0\right)m .
\]
With $\epsilon_0=\frac{1}{4\cdot 1016}$, the two bounds differ by $\frac{1}{4\cdot 1016}m$, giving a constant multiplicative gap. 
Thus, a PTAS for minimum-social-cost PROPm orientations for chores would distinguish the two cases of the \textsc{Max-2P2N-3SAT} gap theorem in polynomial time, unless $\text{P} = \text{NP}$.
\end{proof}

\begin{restatable}{theorem}{thmChoresStarNphard}\label{thm:chores-star-nphard}
For additive chores costs, computing a social cost-minimizing PROPm orientation is
NP-hard even when the graph is a star.
\end{restatable}

\begin{proof}
Apply Lemma~\ref{lem:goods-chores-duality} to the simple-star reduction from
\textsc{Subset Sum}. Let $I^g$ be the goods instance in that reduction, let
$K$ be the welfare threshold, and let
\[
    T=\sum_{e=\{i,j\}\in E}\big(v_i(e)+v_j(e)\big).
\]
Construct the chores instance $I$ on the same star by setting $c_i(e)=v_i(e)$.
For every goods orientation $\pi$, Lemma~\ref{lem:goods-chores-duality} gives
$\mathsf{SW}_{I^g}(\pi)+\mathsf{SC}_{I}(\bar\pi)=T$. Hence $I^g$ admits a
PROPm orientation with welfare at least $K$ if and only if $I$ admits a PROPm
orientation with social cost at most $T-K$. The reduction preserves the star
graph, so the chores problem is NP-hard on stars.
\end{proof}

\begin{restatable}{proposition}{propChoresPropmPoIncompatible}\label{prop:chores-propm-po-incompatible}
PROPm orientations for chores may not be compatible with Pareto optimality,
even when there are only two agents.
\end{restatable}

\begin{proof}
Apply Lemma~\ref{lem:goods-chores-duality} to the two-agent goods instance in
Proposition~\ref{prop:propmPoIncompatibleTwoAgents}. Let $I^g$ be that goods
instance, and construct the chores instance $I$ on the same multigraph by
setting $c_i(e)=v_i(e)$ for every agent $i$ and incident edge $e$.

If $I$ admitted an orientation that is both PROPm and Pareto optimal, then
flipping it would give an orientation of $I^g$ that is both PROPm and Pareto
optimal, by Lemma~\ref{lem:goods-chores-duality}. This contradicts
Proposition~\ref{prop:propmPoIncompatibleTwoAgents}.
\end{proof}

\begin{restatable}{theorem}{thmChoresPropmPoHard}\label{thm:chores-propm-po-hard}
For additive chores costs, it is NP-hard to decide whether there exists an
orientation that is both PROPm and Pareto optimal, even when all costs belong
to $\{0,1,2,4\}$.
\end{restatable}

\begin{proof}
Reduce from the corresponding goods problem. Given a goods instance
$I^g=(G,\mathbf v)$, construct the chores instance $I=(G,\mathbf c)$ on the
same graph by setting $c_i(e)=v_i(e)$ for every incident edge. By
Lemma~\ref{lem:goods-chores-duality}, an orientation $\pi$ is both PROPm and
Pareto optimal in $I^g$ if and only if the flipped orientation $\bar\pi$ is
both PROPm and Pareto optimal in $I$. The construction preserves the graph and
the set of numerical values. Since the goods problem is NP-hard for values in
$\{0,1,2,4\}$, the chores problem is NP-hard for costs in $\{0,1,2,4\}$.
\end{proof}

\begin{restatable}{theorem}{thmChoresPropavgNonexist}\label{thm:chores-propavg-nonexist}
PROPavg orientations for chores may not exist, even for simple graphs and
binary costs.
\end{restatable}

\Cref{thm:chores-propavg-nonexist} follows from Proposition~\ref{prop:prop-avgNotExist} and Lemma~\ref{lem:goods-chores-duality}.


\begin{restatable}{theorem}{thmChoresPropavgNpc}\label{thm:chores-propavg-npc}
Deciding whether a PROPavg orientation exists for chores is NP-complete,
even for simple graphs with costs in $\{0,1,3\}$.
\end{restatable}

\begin{proof}
Membership in NP is immediate. For hardness, apply the transformation in Lemma~\ref{lem:goods-chores-duality} to the goods instance in the NP-completeness reduction for
PROPavg. By Lemma~\ref{lem:goods-chores-duality}, the goods instance admits a
PROPavg orientation if and only if the flipped chores instance admits one. The
construction preserves simplicity of the graph and preserves all numerical
values.
\end{proof}

\vspace{-0.8em}

\begin{restatable}{theorem}{thmChoresStarFptas}\label{thm:chores-star-fptas}
For additive chores costs, there is an FPTAS for computing a minimum-social-cost
PROPm orientation when the graph is a star.
\end{restatable}

\begin{proof}
Let the star have center $c$ and leaves $\ell_1,\ldots,\ell_m$, and let
$e_i=\{c,\ell_i\}$. Write $a_i=c_c(e_i)$ and
$b_i=c_{\ell_i}(e_i)$, and let $A=\sum_i a_i$.

Every leaf is PROPm in every orientation. If a leaf receives its only edge, it
can remove that edge in the PROPm correction. If it does not receive the edge,
its cost is zero. Thus only the center needs to be checked.

Let $S$ be the set of edges assigned to the center. The center is PROPm if and
only if either $S=\emptyset$, or
\[
    \sum_{i\in S}a_i-\max_{i\in S}a_i\le \frac{1}{2}A .
\]
The social cost of the orientation is
$\sum_{i\in S}a_i+\sum_{i\notin S}b_i$.

We guess an edge $h\in S$ that has maximum cost for the center among edges in
$S$. Once $h$ is fixed, every other edge assigned to the center has cost at
most $a_h$ for the center, and the PROPm condition becomes
\[
    \sum_{i\in S\setminus\{h\}}a_i\le \frac{1}{2}A .
\]
For this guess, the remaining problem is to choose
$T\subseteq\{i\neq h:a_i\le a_h\}$ with $\sum_{i\in T}a_i \le \frac{1}{2}A$ so as to
minimize
\[
    a_h+\sum_{i\in T}a_i+\sum_{i\notin T\cup\{h\}}b_i .
\]
We also consider the orientation that assigns all edges to leaves.

For each fixed $h$, regard every edge as having two assignment choices, with costs $a_i$ and $b_i$. Force $h$ to the center and each $i\ne h$ with $a_i>a_h$ to its leaf. The resource consumed by assigning $i\ne h$ to the center is $a_i$; all other choices consume zero resource. The capacity is $\frac{1}{2}A$.

First test for a feasible zero-cost assignment: among the allowed zero-cost choices for each edge, choose one of minimum resource consumption and check the capacity. If one exists, return it. Otherwise let $z_{\min}$ and $z_{\max}$ be the smallest and largest positive endpoint costs. Try bounds $\Gamma=2^t z_{\min}$, starting at $t=0$ and ending at the first bound at least $m z_{\max}$. For $0<\epsilon<1$, set $\kappa=\frac{\epsilon}{2m}\Gamma$ and round each assignment cost $z$ down to $\lfloor \frac{z}{\kappa}\rfloor$. A dynamic program processes the edges and stores, for each total rounded cost up to $\lfloor\frac{\Gamma}{\kappa}\rfloor$, the minimum resource consumption and a corresponding assignment. Return the least rounded-cost feasible assignment, if any.

For a guess $h$ containing an optimal positive-cost solution of cost $C_h$, one of the tested bounds satisfies $C_h\le\Gamma<2C_h$. The optimum fits within the rounded-cost cap, and rounding introduces total error less than $m\kappa$. The returned actual cost is at most $C_h+m\kappa\le(1+\epsilon)C_h$. Taking the best over all bounds and guesses, including the all-to-leaves orientation, proves the claim. There are $O(\frac{m}{\epsilon})$ rounded-cost states and $O(1+\log(\frac{m}{z_{\min}}z_{\max}))$ bounds. The running time is polynomial in the input length and $\frac{1}{\epsilon}$.
\end{proof}

We finally turn to mixed items. The definitions above are chosen so that the
pure-goods and pure-chores notions are obtained as special cases.

\begin{restatable}{theorem}{thmMixedPropmExists}\label{thm:mixed-propm-exists}
Every mixed instance admits a PROPm orientation. Moreover, such an orientation
can be computed in polynomial time.
\end{restatable}

\begin{proof}
We reduce the mixed instance to a goods instance. First orient every edge that
is a good for one endpoint and a chore for the other toward the endpoint for
whom it is a good. This contributes nonnegatively to both endpoints relative to
their proportional shares: the endpoint receiving a good gains value, while the
other endpoint avoids receiving a chore.

It remains to orient edges that are mutual goods or mutual chores. Construct an
auxiliary goods instance on these remaining edges. For a mutual-good edge
$e$, set $\tilde v_i(e)=v_i(e)$ for each endpoint $i$. For a mutual-chore edge
$e$, set $\tilde v_i(e)=c_i(e)$ for each endpoint $i$. By
Theorem~\ref{thm:propmExists}, compute a PROPm orientation $\tilde\pi$ of this
auxiliary goods instance.

Map $\tilde\pi$ back to the mixed instance as follows. Mutual-good edges keep
their orientation. Mutual-chore edges are flipped. Together with the edges
oriented in the first step, this gives an orientation $\pi$ of the mixed
instance.

Fix an agent $i$. The first-step edges contribute nonnegatively to
$u_i(\pi_i)-\frac{1}{2}u_i(E_i)$. On every remaining edge, the contribution of the
mixed orientation $\pi$ is exactly the same as the contribution of
$\tilde\pi$ in the auxiliary goods instance. Therefore, if $i$ reaches its
proportional share in the auxiliary goods instance, then it also reaches its
proportional share in the mixed instance.
Otherwise, since $\tilde\pi$ is PROPm in the auxiliary goods instance, there is
a neighbor $j$ such that
\[
    \tilde v_i(\tilde\pi_i)+
    \min_{e\in E_{ij}\cap\tilde\pi_j}\tilde v_i(e)
    \ge \frac{1}{2} \tilde v_i(E_i).
\]
The set $E_{ij}\cap\tilde\pi_j$ corresponds in the mixed instance to goods that
$i$ loses to $j$, chores that $i$ receives from $j$, or both. Hence, the
auxiliary correction is no larger than
$\max\{M_i^g(\pi),M_i^c(\pi)\}$. Thus agent $i$ satisfies
Definition~\ref{def:mixed-propm}. Since $i$ was arbitrary, $\pi$ is PROPm.
\end{proof}

\section{Missing Proofs from \Cref{sec:PROPmExists}}

\subsection{Proof of \Cref{lem:partition}} \label{sec:lem-partitionAppendix}

\lemPartition*

\begin{proof}
The algorithm removes an edge only when the remaining set is still an
$i$-majority set. Hence, it maintains the invariant
\[
    v_i(P)\ge \frac{1}{2} v_i(E_{ij}),
\]
and the returned set is an $i$-majority set.

It remains to verify the minimality. Let $P^{(\ell)}$ be the set immediately after
the algorithm has processed $e_\ell$ for all $\ell$. Consider an edge $e_k$ that belongs to the
returned set $P$. Then the algorithm did not remove $e_k$, and 
\[
    v_i(P^{(k-1)}\setminus\{e_k\})<\frac{1}{2} v_i(E_{ij}).
\]
Since the set $P$ only shrinks during the algorithm and valuations are
nonnegative, we have
\[
    v_i(P_{ij}\setminus\{e_k\})
    \le
    v_i(P^{(k-1)}\setminus\{e_k\})
    <
    \frac{1}{2} v_i(E_{ij}).
\]
Thus $P_{ij}$ is inclusion-minimal. Finally, since $P_{ij}$ is an
$i$-majority set, the only remaining distinction is whether
$v_j(P_{ij})$ is below or at least $\frac{1}{2}v_j(E_{ij})$, which yields (C1) and
(C2).
\end{proof}

\subsection{Proof of \Cref{lem:greedy-walk}} \label{sec:lem-greedy-walkAppendix}

\lemGreedyWalk*

\begin{proof}
Fix a vertex $u_i$ and consider the subsequence of the greedy walk consisting
only of the times at which an edge incident to $u_i$ is processed. Such an
edge is either \emph{won} by $u_i$, meaning that the walk is currently at $u_i$ and
the edge is oriented toward $u_i$, or \emph{lost} by $u_i$, meaning that the walk is
currently at the other endpoint and the edge is oriented away from $u_i$.

We compare lost edges with won edges as follows. Consider the subsequence of processed edges that are incident to $u_i$. Two lost edges cannot appear consecutively. 
Indeed, suppose an edge $\{u_i,u_j\}$ is lost by $u_i$, so it is oriented toward $u_j$. After processing this edge, the algorithm sets the current vertex to $u_i$. Therefore, if the next processed edge incident to $u_i$ exists, it is chosen while the current vertex is $u_i$ and hence is oriented toward $u_i$.

Hence, except possibly for the first lost edge in this subsequence, every lost
edge is immediately preceded by a distinct won edge. Pair each such lost edge
$e$ with the won edge $e'$ that immediately precedes it. When $e'$ was chosen,
the edge $e$ was still unprocessed and incident to $u_i$. Since the greedy walk
chooses a remaining incident edge of maximum weight from the perspective
of $u_i$,
\[
    w_i(e')\ge w_i(e).
\]
Thus, every lost edge except possibly the first lost edge is paired with a distinct won edge of at least the same weight. 
The first lost edge, if it exists, has weight at most $\max_{e\in L_i(\Pi)}w_i(e)$. Therefore,
\[
    \sum_{e\in L_i(\Pi)}w_i(e)
    \le
    \sum_{e\in W_i(\Pi)}w_i(e)
    +
    \max_{e\in L_i(\Pi)}w_i(e),
\]
which proves the lemma.
\end{proof}

\section{Missing Proofs from Section \ref{sec:priceOfPROPm}} 

\subsection{Proof of \Cref{prop:auxiliary-graph}} \label{sec:auxiliary-graphAppendix}

\propAuxiliaryGraph*

\begin{proof}
Each edge of $H$ is created from a pair of distinct virtual edges in the
sorted list of some virtual vertex $u_i$. Hence, it connects two distinct
vertices of $H$, and so $H$ has no self-loop. Moreover, every edge of $H$ is
created in exactly one sorted list and therefore receives exactly one label.
We next show that $H$ has no parallel edges. Suppose, toward a contradiction,
that two vertices $x_e$ and $x_{e'}$ are joined by two distinct edges in $H$.
Then the virtual edges $e$ and $e'$ must have been paired in the sorted list
of two distinct endpoints in $\mathcal G$. Since a virtual edge has exactly
two endpoints, this implies that $e$ and $e'$ have the same two endpoints in
$\mathcal G$. This contradicts the fact that $\mathcal G$ is simple. Thus
$H$ is simple.

Finally, fix a vertex $x_e$ of $H$. The corresponding virtual edge $e$ has
two endpoints in $\mathcal G$. At each endpoint, $e$ appears once in the
sorted list and can be paired with at most one other edge. Hence $x_e$ is
incident to at most one edge of $H$ from each endpoint of $e$ and therefore
has degree at most $2$.
Since every vertex of $H$ has degree at most $2$, each connected component of
$H$ is a path, a cycle, or an isolated vertex. Orienting each path
consistently and each cycle cyclically gives every vertex at most one incoming
edge.
\end{proof}

\subsection{Proof of \Cref{prop:priceLowerBound}} \label{sec:priceLowerBoundAppendix}

\propPriceLowerBound*

\begin{proof}
Fix an even integer $d=2k$ with $k\ge 1$ and fix a number $H>1$. Consider a
star with center $0$ and leaves $1,\ldots,d$. For each edge
$e_\ell=\{0,\ell\}$, set $v_0(e_\ell)=1$ and $v_\ell(e_\ell)=H$.
The unconstrained welfare-maximizing orientation assigns every edge to its
leaf and therefore has welfare $dH=2kH$.
We now characterize the PROPm orientations. First, every leaf satisfies PROPm regardless of the orientation of its incident edge.
Suppose the center receives $t$ edges. If $t<d$, then the center loses at least one edge. Since all
edges have value $1$ for the center, its PROPm correction is $1$. Hence the
PROPm condition for the center requires $t+1\ge \frac{1}{2}d =k$, so $t\ge k-1$. Thus, any PROPm orientation must assign at least $k-1$ edges to the center.
Since $H>1$, welfare is maximized by assigning as few edges as possible to the
center, namely exactly $k-1$ edges. The PROPm orientation in this instance has welfare at most
$(k-1)+(k+1)H$.
Consequently, the price of PROPm is at least
\[
    \frac{2}{(k-1)+(k+1)H}kH.
\]
Letting $H\to\infty$ gives $\frac{2}{k+1}k$, and the ratio approaches $2$ as $k\to\infty$.
\end{proof}

\subsection{Proof of \Cref{lem:price-propm}} \label{sec:lem-price-propmAppendix}

\lemPricePropm*

\begin{proof}
For an orientation $\Gamma\in\{\Pi,\widehat{\Pi}\}$ of the virtual graph
$\mathcal G$, recall that $W_i(\Gamma)$ denotes the set of virtual edges
incident to $u_i$ that are oriented toward $u_i$, and $L_i(\Gamma)$ denotes
the set of virtual edges incident to $u_i$ that are oriented away from $u_i$.
We first prove that for every $u_i$ and every
$\Gamma\in\{\Pi,\widehat{\Pi}\}$,
\[
    \sum_{e\in L_i(\Gamma)}w_i(e)
    \le
    \sum_{e\in W_i(\Gamma)}w_i(e)
    +
    \max_{e\in L_i(\Gamma)}w_i(e),
\]
where the maximum is defined as $0$ if $L_i(\Gamma)=\emptyset$.

Fix a virtual vertex $u_i$, and let $e_1,\ldots,e_d$ be the virtual edges
incident to $u_i$, sorted so that
\[
    w_i(e_1)\ge w_i(e_2)\ge\cdots\ge w_i(e_d).
\]
By construction, for every $r\le \lfloor\frac{1}{2}d\rfloor$, the auxiliary graph
$H$ contains an edge between $x_{e_{2r-1}}$ and $x_{e_{2r}}$ labeled by $i$.
Exactly one endpoint receives this auxiliary edge. The other endpoint either receives an edge with the other label or is a path endpoint; in either case, it is not assigned to $u_i$. Therefore,
by the conversion rule, exactly one of $e_{2r-1}$ and $e_{2r}$ is oriented
toward $u_i$ in $\Pi$, and the other is oriented away from $u_i$. The same
one-won-one-lost property holds in $\widehat{\Pi}$, since reversing every
virtual edge also reverses the status of both edges in each pair.
We now prove the above displayed inequality. In the first pair
$\{e_1,e_2\}$, the lost edge, if any, is left uncharged. For each later pair
$\{e_{2r-1},e_{2r}\}$ with $r\ge 2$, charge the lost edge in this pair to the
edge won in the previous pair $\{e_{2r-3},e_{2r-2}\}$. This charging is
injective. Moreover, the won edge in the previous pair has weight at least
$w_i(e_{2r-2})$, while the lost edge in the current pair has weight at most
$w_i(e_{2r-1})$. Since the list is sorted,
\[
    w_i(e_{2r-2})\ge w_i(e_{2r-1}),
\]
so each charged lost edge is charged to a distinct won edge of at least the
same weight. The cases $d\le 1$ are immediate. If $d\ge 3$ is odd and the unpaired edge $e_d$ is lost, then it can be
charged to the won edge in the last pair, whose weight is at least
$w_i(e_{d-1})\ge w_i(e_d)$. Hence, all lost edges except possibly one are
charged to distinct won edges of at least the same weight. The only possible
uncharged lost edge has weight at most
$\max_{e\in L_i(\Gamma)}w_i(e)$, proving the inequality.
Note that the inequality is the same as the inequality in the statement of Lemma~\ref{lem:greedy-walk}.
Then, by arguments almost the same as those in the proofs of Lemma~\ref{lem:reduced} and of Theorem~\ref{thm:propmExists}, one can easily verify that using either $\Pi$ or $\widehat{\Pi}$ in the Steps~\ref{step::alg::main-final-start}--\ref{step::alg::main-final-end} of Algorithm~\ref{alg:main} gives a PROPm
orientation.
\end{proof}

\section{Missing Proofs from Section \ref{sec:optimization:welfare}}
\label{app:proofs:optimization}

\subsection{Proof of Theorem \ref{thm:optimization:welfare:multiple}}

\thmOptWelfareMultiple*
\begin{proof}
We give a reduction from \textsc{Max-2P2N-3SAT}. In this problem, each clause contains exactly three literals, and each variable appears exactly twice as a positive literal and exactly twice as a negative literal; the objective is
to find a truth assignment that maximizes the number of satisfied clauses.
We use the hardness result of Berman et al. \cite{berman2005approximation}: for every fixed $\epsilon>0$, it is NP-hard to distinguish between an instance in which at least a $(1-\epsilon)$ fraction of the clauses can be satisfied and an instance in which every assignment satisfies at most a $(\frac{1015}{1016}+\epsilon)$ fraction of the clauses. Set $\epsilon_0=\frac{1}{4\cdot 1016}$.

Let $\varphi$ be an instance with variables $x_1,\ldots,x_n$ and clauses $C_1,\ldots,C_m$, where each clause has exactly three literals and each variable appears exactly twice positively and twice negatively. Thus $4n=3m$.
For each variable $x_i$, create two literal agents, denoted by $x_i$ and $\bar{x}_i$. For each clause $C_j$, create a clause agent $c_j$ and a dummy agent $d_j$. For a literal $\ell$, let $a(\ell)$ denote its corresponding literal agent. The constructed instance has $2n+2m$ agents.
For each variable $x_i$, create three parallel edges between $x_i$ and $\bar{x}_i$, denoted by $s_i,p_i^+,p_i^-$. Set
\[
v_{x_i}(s_i)=v_{\bar{x}_i}(s_i)=10,\quad v_{x_i}(p_i^+)=v_{\bar{x}_i}(p_i^-)=5, \quad v_{\bar{x}_i}(p_i^+) = v_{x_i}(p_i^-) = 0.
\]
The edge $s_i$ refers to the \emph{switch} edge. For a literal agent $a$, let $a^{\mathrm{opp}}$ be the opposite literal agent of the same variable, and let $p^0(a)$ be the edge among $p_i^+,p_i^-$ that has value $0$ for $a$.
For each clause $C_j=(\ell_j^1\vee \ell_j^2\vee \ell_j^3)$ and each $t\in[3]$, create three parallel edges $g_j^t,q_j^t,r_j^t$ between $a(\ell_j^t)$ and $c_j$. The edge $g^t_j$ refers to the \emph{main occurrence} edge, so each $c_j$ has three main occurrence edges.
Set
\begin{align*}
&\big(v_{a(\ell_j^t)}(g_j^t),v_{c_j}(g_j^t)\big)=(2,3), \\
&\big(v_{a(\ell_j^t)}(q_j^t),v_{c_j}(q_j^t)\big)=(0,2),\\
&\big(v_{a(\ell_j^t)}(r_j^t),v_{c_j}(r_j^t)\big)=(2,0).
\end{align*}
Finally, for each clause $C_j$, create an edge $h_j=(c_j,d_j)$ with
$v_{c_j}(h_j)=2$ and $v_{d_j}(h_j)=0$. The construction has $3n+10m$ edges, and all valuations are additive. The gadgets and valuations are illustrated in Figure~\ref{fig:propM-welfare-gadgets}.

\begin{figure}[t]
\centering
\begin{tikzpicture}[
    agent/.style={circle,draw,minimum size=8mm,inner sep=1pt,font=\small},
    litagent/.style={circle,draw,minimum size=13mm,inner sep=1pt,font=\small},
    val/.style={font=\scriptsize,inner sep=1pt,fill=white},
    elab/.style={font=\scriptsize,inner sep=1pt,fill=white},
    panel/.style={font=\small}
]

\node[agent] (xi) at (0,0) {$x_i$};
\node[agent] (nxi) at (3,0) {$\bar{x}_i$};

\draw (xi) to[bend left=35] node[elab,above] {$s_i$} (nxi);
\node[val] at ($(xi)!0.22!(nxi)+(0,0.55)$) {$10$};
\node[val] at ($(xi)!0.78!(nxi)+(0,0.55)$) {$10$};

\draw (xi) -- node[elab,above] {$p_i^+$} (nxi);
\node[val] at ($(xi)!0.22!(nxi)+(0,0.15)$) {$5$};
\node[val] at ($(xi)!0.78!(nxi)+(0,0.15)$) {$0$};

\draw (xi) to[bend right=35] node[elab,below] {$p_i^-$} (nxi);
\node[val] at ($(xi)!0.22!(nxi)+(0,-0.55)$) {$0$};
\node[val] at ($(xi)!0.78!(nxi)+(0,-0.55)$) {$5$};

\node[panel] at (1.5,-1.3) {(a) Variable gadget};

\node[litagent] (lit) at (5.4,0) {$a(\ell_j^t)$};
\node[agent] (cj) at (8.8,0) {$c_j$};
\node[agent] (dj) at (10.6,0) {$d_j$};

\draw (lit) to[bend left=35] node[elab,above] {$g_j^t$} (cj);
\node[val] at ($(lit)!0.22!(cj)+(0,0.55)$) {$2$};
\node[val] at ($(lit)!0.78!(cj)+(0,0.55)$) {$3$};

\draw (lit) -- node[elab,above] {$q_j^t$} (cj);
\node[val] at ($(lit)!0.22!(cj)+(0,0.15)$) {$0$};
\node[val] at ($(lit)!0.78!(cj)+(0,0.15)$) {$2$};

\draw (lit) to[bend right=35] node[elab,below] {$r_j^t$} (cj);
\node[val] at ($(lit)!0.22!(cj)+(0,-0.55)$) {$2$};
\node[val] at ($(lit)!0.78!(cj)+(0,-0.55)$) {$0$};

\draw (cj) -- node[elab,above] {$h_j$} (dj);
\node[val] at ($(cj)!0.28!(dj)+(0,0.18)$) {$2$};
\node[val] at ($(cj)!0.72!(dj)+(0,0.18)$) {$0$};

\node[panel] at (8.0,-1.3) {(b) Occurrence and dummy gadget};

\end{tikzpicture}
\caption{The gadgets used in the APX-hardness reduction. The label closer to
an endpoint is the value of the edge for that endpoint. For each clause
$C_j=(\ell_j^1\vee \ell_j^2\vee \ell_j^3)$, the construction creates three
copies of the occurrence gadget in panel (b), one for each $t\in[3]$, all
sharing the same clause agent $c_j$ and the same dummy edge $h_j=(c_j,d_j)$.}
\label{fig:propM-welfare-gadgets}
\end{figure}
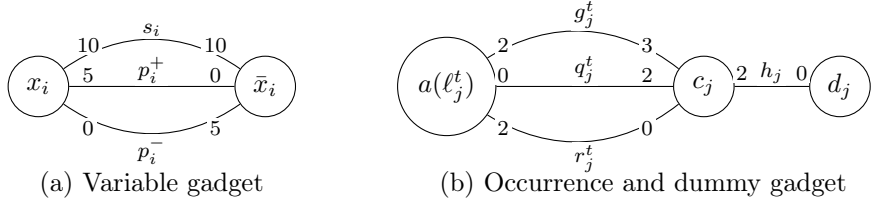

Each literal agent has a total value of $10+5+2(2+2)=23$ and hence a proportional share $\frac{23}{2}$. Each clause agent has a total value $3(3+2)+2=17$ and hence a proportional share $\frac{17}{2}$. Each dummy agent has a total value of $0$ and is henceforth always PROPm.

For a truth assignment $\tau$, define its \emph{canonical orientation} $\pi^\tau$ as follows. For each $i\in [n]$, assign $s_i$ to $x_i$ if $x_i$ is true and to $\bar{x}_i$ otherwise. 
Assign $p_i^+$ to $x_i$ and $p_i^-$ to $\bar{x}_i$. Moreover, for every $j\in [m]$ and $t\in[3]$, assign every $q_j^t$ to $c_j$, every $r_j^t$ to $a(\ell_j^t)$, and every $h_j$ to $c_j$. Finally, assign $g_j^t$ to $c_j$ if $\ell_j^t$ is true under $\tau$ and to $a(\ell_j^t)$ otherwise.
Since every assignment makes exactly $2n$ literal occurrences true and $2n$ literal occurrences false, the welfare of $\pi^\tau$ is always equal to
\begin{align*}
W^\star & =10n+10n+3(2n)+2(2n)+2(3m)+2(3m)+2m \\
& =30n+14m \\
&=\frac{73}{2}m,
\end{align*}
independent of $\tau$. The last equality uses $4n=3m$.

We first prove the completeness of the reduction. Suppose that there exists a truth assignment $\tau$ that leaves at most $\epsilon_0 m$ clauses unsatisfied. We show that there is a PROPm orientation with welfare at least $W^\star-2\epsilon_0 m$.
More generally, let $u$ be the number of clauses left unsatisfied by $\tau$.
Starting from $\pi^\tau$, every satisfied clause agent receives $h_j$, its three $q$-edges, and at least one main occurrence edge, for a value at least $2+3\cdot 2+3=11>\frac{17}{2}$. An unsatisfied clause agent receives only $h_j$ and its three $q$-edges, for a value of $2+3\cdot 2=8$.
For each unsatisfied clause, choose one of its three literals arbitrarily and move the corresponding edge $g_j^t$ from the literal agent to $c_j$. If a false literal agent $a$ is chosen once, say for occurrence $(j,t)$, also move $q_j^t$ from $c_j$ to $a$. Then $a$ has own value $11$, and the bundle assigned to $c_j$ among the edges between $a$ and $c_j$ consists only of $g_j^t$, which has value $2$ for $a$. Thus $a$ remains PROPm. The repaired clause agent has value at least $2+2\cdot 2+3=9>\frac{17}{2}$.

If a false literal agent $a$ is chosen twice, move both corresponding $g$-edges to the two clause agents and move $p^0(a)$ from $a^{\mathrm{opp}}$ to $a$. Then $a$ has value $5+2\cdot 2=9$ but receives a PROPm correction of value $10$ from the switch edge assigned to $a^{\mathrm{opp}}$. Hence $a$ remains PROPm. The opposite literal agent loses an edge of value $5$ but still receives the switch edge and its two $r$-edges, for value at least $10+2\cdot 2=14>\frac{23}{2}$. Each repaired clause agent again has value at least $9$.

It remains to bound the welfare loss. If a false literal is chosen once, moving its $g$-edge to the clause agent gains $1$ welfare, while moving the corresponding $q$-edge to the literal agent loses $2$, for a net loss of $1$. 
If a false literal is chosen twice, the two moved $g$-edges gain $2$ welfare, while moving $p^0(a)$ loses $5$, for a net loss of $3$. 
If $n_1$ and $n_2$ are the numbers of false literals chosen once and twice, then $n_1 + 2n_2 = u$, and the total loss is $n_1 + 3n_2 \leq 2u$. 
Therefore, there is a PROPm orientation with welfare at least $W^\star-2u$.

In particular, some assignments leave at most $\epsilon_0 m$ clauses unsatisfied. 
Hence,
\[
\mathsf{OPT}_{\mathrm{PROPm}}\ge W^\star-2\epsilon_0 m.
\]

We next prove soundness. 
Suppose that every truth assignment leaves at least $(\frac{1}{1016}-\epsilon_0)m$ clauses unsatisfied. 
We show that every PROPm orientation has welfare at most
\[
    W^\star-\left(\frac{1}{1016}-\epsilon_0\right)m.
\]

Let $\pi$ be any PROPm orientation. Decode a truth assignment $\tau(\pi)$ by setting $x_i$ to true if and only if $s_i$ is assigned to $x_i$ in $\pi$. Let $\widehat{\pi}=\pi^{\tau(\pi)}$ be the canonical orientation induced by this assignment, with the same switch-edge allocation as $\pi$. Then $\mathsf{SW}(\widehat{\pi})=W^\star$.

Let $U(\tau(\pi))$ be the number of clauses unsatisfied by $\tau(\pi)$. We show that
\[
\mathsf{SW}(\pi)\le W^\star-U(\tau(\pi)).
\]
Fix a literal agent $a$ that is false under $\tau(\pi)$. Let $G_a$ be the number of its two $g$-edges assigned to clause agents in $\pi$, let $R_a$ be the number of its two $r$-edges assigned to clause agents in $\pi$, let $Q_a$ be the number of its two $q$-edges assigned to $a$ in $\pi$, and let $P_a\in\{0,1\}$ indicate whether $p^0(a)$ is assigned to $a$ in $\pi$. Define $L_a=G_a+R_a$.
Compared with $\widehat{\pi}$, these counted deviations change welfare by at most
\[
\Delta_a=G_a-2R_a-2Q_a-5P_a.
\]
We claim that $\Delta_a\le -L_a$. If $L_a=0$, this is immediate. Suppose $L_a=1$. In $\widehat{\pi}$, agent $a$ has value $5+2\cdot 2+2\cdot 2=13$. Losing one incident $g$- or $r$-edge leaves it with value at most $11<\frac{23}{2}$. 
If $P_a=Q_a=0$, then PROPm correction is zero: the opposite literal agent receives $p^0(a)$, and each adjacent clause agent receives the corresponding $q$-edge. Thus, $a$ is not PROPm, a contradiction. Therefore, PROPm implies $P_a+Q_a\ge 1$, and so $\Delta_a\le 1-2=-1=-L_a$.

Now suppose $L_a\ge 2$. If $P_a=0$, then the opposite literal agent receives $p^0(a)$, which has value
$0$ for $a$, so the opposite literal agent contributes no positive PROPm
correction for $a$.
Moreover, for any adjacent clause agent, the corresponding PROPm correction term is at most $2$. 
Since $a$ loses at least two incident edges of value $2$, its own value is at most $9$, and even with the best possible witness, it reaches at most $11<\frac{23}{2}$, contradicting PROPm. 
Hence $P_a=1$. Using $G_a \leq 2$,
\begin{align*}
\Delta_a &\le G_a-2R_a-5 \\
&=3G_a-2L_a-5 \\
&\le 1-2L_a \\
&\le -L_a.
\end{align*}
This proves the claim.

Consider now a clause $C_j$ that is unsatisfied by $\tau(\pi)$. All three of its literals are false. 
We claim that in $\pi$, at least one of the six edges $g_j^1,g_j^2,g_j^3,r_j^1,r_j^2,r_j^3$ is assigned to $c_j$. 
Otherwise, $c_j$ receives no $g$-edge and no $r$-edge. 
If $h_j$ is assigned to $c_j$, then its value is at most $2+3\cdot 2=8$. If $h_j$ is assigned to $d_j$, then its own value is at most $6$, and the edge $h_j$ can contribute a correction of at most $2$. 
In both cases, any correction obtained from a literal neighbor is $0$, because the corresponding edge $r_j^t$ is assigned to that literal agent and has value $0$ for $c_j$. 
Thus $c_j$ would violate PROPm. Therefore, every unsatisfied clause contributes at least one edge counted in $L_a$ for some false literal agent $a$, and hence
\[
\sum_{a\text{ false}} L_a\ge U(\tau(\pi)).
\]

All remaining edges are assigned in $\widehat{\pi}$ to welfare-maximizing endpoints or are switch edges whose orientation is the same in $\pi$ and $\widehat{\pi}$. 
Therefore, their welfare contribution in $\pi$ is no larger than in $\widehat{\pi}$. 
Consequently,
\begin{align*}
    \mathsf{SW}(\pi)-\mathsf{SW}(\widehat{\pi}) 
    &\le \sum_{a\text{ false}}\Delta_a \\
    &\le -\sum_{a\text{ false}}L_a \\
    &\le -U(\tau(\pi)).
\end{align*}
This proves the desired comparison.

In the soundness case, every truth assignment leaves at least $(\frac{1}{1016}-\epsilon_0)m$ clauses unsatisfied. Hence, every PROPm orientation has welfare at most
\[
W^\star-\left(\frac{1}{1016}-\epsilon_0\right)m.
\]

The completeness value is at least $W^\star-2\epsilon_0m$, while the soundness value is at most $W^\star-(\frac{1}{1016}-\epsilon_0)m$. Since $\epsilon_0=\frac{1}{4\cdot 1016}$, the additive gap between these two bounds is
\[
\left(\frac{1}{1016}-3\epsilon_0\right)m=\frac{1}{4\cdot 1016}m.
\]
Because $W^\star= \frac{73}{2}m$, this is a constant multiplicative gap. 
Therefore, a PTAS for computing a maximum-welfare PROPm orientation would distinguish the two cases of the \textsc{Max-2P2N-3SAT} gap theorem in polynomial time, unless $P=NP$. 
Thus, the problem is APX-hard.
\end{proof}

\subsection{Proof of Theorem \ref{thm:optimization:welfare:simple}}

\thmOptWelfareSimple*
\begin{proof}
We reduce from \textsc{Subset Sum}. In this problem, the input consists of positive integers $a_1,\ldots,a_m$ and a target $B$, and the goal is to decide whether there exists a subset $S\subseteq[m]$ such that $\sum_{i\in S} a_i=B$. Let $W=\sum_{i=1}^m a_i$. We may assume that $0<B\le W$, since otherwise the instance is trivial.

We now construct the orientation instance. Create one agent $c$, which will be the center of the star. For each integer $a_i$, create an item agent $x_i$ and an edge $e_i=\{c,x_i\}$. We also create a constant number of dummy agents, as defined below.
If $2B\le W$, create two dummy agents $y_1,y_2$ and two dummy edges $f_1=\{c,y_1\}$ and $f_2=\{c,y_2\}$. Set $r_1=W$ and $r_2=2B$.
If $2B>W$, create three dummy agents $y_1,y_2,y_3$ and three dummy edges $f_1=\{c,y_1\}$, $f_2=\{c,y_2\}$, and $f_3=\{c,y_3\}$. Set $r_1=W$, $r_2=W$, and $r_3=2B-W$.
Let $q$ be the number of dummy agents created. In both cases, $q\in\{2,3\}$,
\[
    \sum_{t=1}^q r_t=W+2B
    \quad\text{and}\quad
    \max_{t\in[q]}\; r_t=W .
\]
Thus, the constructed graph is a star with $m+q+1$ vertices and $m+q$ edges.

The valuations are additive. For each item edge $e_i$, set $v_c(e_i)=a_i$ and $v_{x_i}(e_i)=2a_i$. 
For each dummy edge $f_t$, set $v_c(f_t)=r_t$ and $v_{y_t}(f_t)=H$, where $H=5W+1$. 
Let $K=qH+2W-B$.
We show that the \textsc{Subset Sum} instance is a yes-instance if and only if the constructed instance admits a PROPm orientation with welfare at least $K$.

First, suppose that there is a set $S\subseteq[m]$ with $\sum_{i\in S}a_i=B$. 
Orient each item edge $e_i$ toward $c$ if $i\in S$ and toward $x_i$ otherwise. 
Orient every dummy edge $f_t$ toward $y_t$.

Every leaf is PROPm under any orientation of a star: if it receives its unique incident edge, it receives its full value, while if it does not, then the single edge assigned to its neighbor gives exactly the PROPm correction.
It remains only to check the center. 
The center's total value for all incident edges is $W+(W+2B)=2W+2B$, so $\mathsf{PROP}_c=W+B$.
In the orientation above, the center receives item edges of total value $B$. 
Since all dummy edges are assigned away from $c$ and the largest dummy edge has value $W$ for $c$, the PROPm correction for $c$ is at least $W$.
Hence, $v_c(\pi_c)+W=B+W=\mathsf{PROP}_c$, so the orientation is PROPm.

Its social welfare is
\begin{align*}
    \sum_{i\in S}a_i+\sum_{i\notin S}2a_i+qH
    & =B+2(W-B)+qH \\
    &=qH+2W-B \\
    &=K .
\end{align*}

Conversely, suppose that there is a PROPm orientation $\pi$ with $\mathrm{SW}(\pi)\ge K$. We first claim that every dummy edge is assigned to its dummy vertex. 
If some dummy edge is assigned to $c$, then the dummy edges contribute at most $(q-1)H+\sum_{t=1}^q r_t\le (q-1)H+3W$ to the social welfare, while the item edges contribute at most $2W$. 
Thus, $\mathrm{SW}(\pi) \leq (q-1)H+5W$. 
But
\begin{align*}
    K-\big((q-1)H+5W\big)
    & =H-3W-B \\ 
    & \ge (5W+1)-3W-W \\
    &=W+1 \\
    &>0,
\end{align*}
contradicting $\mathsf{SW}(\pi)\ge K$. Therefore, all dummy edges are assigned to dummy vertices.

Let $S=\{i\in[m]: e_i \text{ is assigned to } c \text{ in } \pi\}$ and let
$X=\sum_{i\in S}a_i$. Since all dummy edges are assigned away from $c$, the
center loses an edge of value $W$, and no lost edge has a value larger than $W$
from the center's perspective. Because the graph is simple, the PROPm
correction for $c$ is therefore exactly $W$. Since $\mathsf{PROP}_c=W+B$ and
$\pi$ is PROPm, we have $X+W\ge W+B$, and hence $X\ge B$.

On the other hand, since all dummy edges are assigned to dummy vertices, the
social welfare of $\pi$ is
\begin{align*}
    qH+\sum_{i\in S}a_i+\sum_{i\notin S}2a_i
    &=qH+X+2(W-X) \\
    &=qH+2W-X .
\end{align*}
The inequality $\mathsf{SW}(\pi) \ge K=qH+2W-B$ implies $X\le B$. Therefore, $X=B$, and the set $S$ is a solution to \textsc{Subset Sum}.
Thus, the constructed instance admits a PROPm orientation with welfare at least $K$ if and only if the original \textsc{Subset Sum} instance is a yes-instance.
Therefore, computing a PROPm orientation with maximum social welfare is NP-hard, even on simple graphs.
\end{proof}

Theorem \ref{thm:optimization:welfare:simple} implies that the problem remains NP-hard, even when the graph is a star. 
In the following, we prove that there is an FPTAS for computing a maximum-welfare
PROPm orientation if the graph is a star.

\begin{restatable}{proposition}{propStarFPTAS}
For additive goods valuations, there is an FPTAS for computing a maximum-welfare
PROPm orientation when the graph is a star.
\end{restatable}

\begin{proof}
Let the star have center $c$ and leaves $\ell_1,\ldots,\ell_m$, and let
$e_i=(c,\ell_i)$. Write $a_i=v_c(e_i)$ and $b_i=v_{\ell_i}(e_i)$, and let
$A=\sum_i a_i$.

First observe that every leaf is always PROPm. Indeed, if $\ell_i$ receives
$e_i$, then it obtains its full value. If it does not receive $e_i$, then the
edge is assigned to $c$, and the PROPm correction term is exactly
$v_{\ell_i}(e_i)$.
Thus, we only need to enforce PROPm for the center. Let $L$ be the set of edges
not assigned to $c$. Since the graph is simple, the center's PROPm condition is
equivalent to
\[
    \sum_{i\in L} a_i-\max_{i\in L} a_i \le \frac{1}{2}A,
\]
with the convention that the condition is trivially satisfied when $L=\emptyset$.

The welfare of an orientation assigning exactly the edges in $L$ to the leaves is
\[
    A+\sum_{i\in L}(b_i-a_i).
\]
Therefore, the problem is to choose a set $L$ maximizing this expression subject
to the above constraint.

We solve this by guessing the edge $h\in L$ with maximum $a_h$.
Once $h$ is fixed, every other edge assigned to a leaf must have value at most $a_h$ for the center. Since $h$ is the edge that provides the PROPm correction, the PROPm condition for the center is equivalent to $\sum_{i\in L\setminus\{h\}} a_i\le \frac{1}{2}A$. 
Hence, the remaining edges assigned to leaves must be chosen from $I_h=\{i\neq h:a_i\le a_h\}$ and have total value at most $\frac{1}{2}A$ for the center.
Thus, for each choice of $h$, solve the knapsack problem:
\begin{align*}
    \max\quad &\sum_{i\in T}(b_i-a_i) \\
    \text{s.t.} \quad
    &T\subseteq I_h, \\
    &\sum_{i\in T}a_i \leq \frac{1}{2}A,
\end{align*}
discarding items with nonpositive profit $b_i-a_i$.
The candidate orientation assigns $e_h$ and all edges in $T$ to their leaves, and all other edges to $c$.
We also consider the orientation that assigns all edges to $c$.

This enumeration is exact if the knapsack subproblems are solved exactly: every feasible PROPm orientation either has $L=\emptyset$ or has some edge $h\in L$ with maximum $a_h$, in which case $L\setminus\{h\}$ is feasible for the
corresponding knapsack instance.

Using the standard FPTAS for knapsack for each guess $h$ gives an FPTAS for the star instance. 
Indeed, the FPTAS returns $T_h$ with profit at least $(1-\epsilon)$ times the optimal knapsack profit, so the resulting welfare is within a factor $(1-\epsilon)$ of the best orientation for that guess. 
Taking the best over all $h$, together with the all-to-center orientation, gives a $(1-\epsilon)$-approximation to the optimum PROPm welfare. 
The running time is polynomial in $m$ and $\frac{1}{\epsilon}$.
\end{proof}

\section{Missing Proofs from Section \ref{sec:PROPm:PO}}
\label{app:PROPm:PO}

\subsection{Proof of Proposition \ref{prop:propmPoIncompatibleTwoAgents}}

\propPropmPoIncompatibleTwoAgents*
\begin{proof}
Consider two agents connected by three parallel edges $p,q,r$, with valuations towards them
given by:
\[
    v_1(p)=2,\quad v_1(q)=0,\quad v_1(r)=1,
    \quad \text{and} \quad
    v_2(p)=2,\quad v_2(q)=1,\quad v_2(r)=0.
\]
Both agents have a proportional share $\frac{3}{2}$. We claim that there are two PROPm orientations between them:
\[
    \pi^{(1)}_1=\{p\},\quad \pi^{(1)}_2=\{q,r\},
    \quad \text{and} \quad
    \pi^{(2)}_1=\{q,r\},\quad \pi^{(2)}_2=\{p\}.
\]
In $\pi^{(1)}$, agent $1$ receives value $2$, while
agent $2$ receives value $1$ and can use the lost edge $p$, worth $2$ to her.
The argument for $\pi^{(2)}$ is symmetric.
We claim that no other orientation is PROPm. If agent $1$ receives $\emptyset$, $\{q\}$,
$\{r\}$, or $\{p,q,r\}$, then one of agents 1 and 2 violates PROPm. If agent $1$ receives $\{p,q\}$, then agent $2$
receives value $0$ and has PROPm correction
$1$. 
If agent $1$ receives $\{p,r\}$, then agent $2$
receives value $1$ and has PROPm correction
$0$. Thus, these orientations are not PROPm.

However, neither of the two PROPm orientations is PO. In $\pi^{(1)}$, moving
$r$ from agent $2$ to agent $1$ results in an orientation that Pareto dominates $\pi^{(1)}$. In $\pi^{(2)}$, moving $q$ from agent $1$ to agent $2$ results in an orientation that Pareto dominates $\pi^{(2)}$.
Therefore, no PROPm orientation is PO in this instance.

The impossibility result can be extended to instances with any $n\geq 3$ agents. One can add $n-2$ dummy agents, each adjacent to agent 1. For each dummy agent $d$ and agent 1, they value edge $\{1,d\}$ at 0.
One can easily verify that no PROPm orientation in the $n$-agent instance is PO.
\end{proof}

\subsection{Proof of Theorem \ref{thm:hard:PO:PROPm}}

\thmHardPOPROPm*
\begin{proof}
We reduce from \textsc{2P2N-3SAT}~\cite{berman2005approximation}. Given an instance with variables
$x_1,\ldots,x_n$ and clauses $C_1,\ldots,C_m$, where each clause has exactly
three distinct literals, and each variable appears exactly twice positively and
twice negatively, we construct a goods instance as follows.

For each variable $x_i$, create two literal agents, also denoted by $x_i$ and
$\bar{x}_i$. For each clause $C_j$, create one clause agent $c_j$. For a literal
$\ell$, let $a(\ell)$ be its corresponding literal agent. Thus the instance has
$2n+m$ agents.
The edges are defined as follows. For each variable $x_i$, create three parallel
edges between $x_i$ and $\bar{x}_i$, denoted by $s_i$, $p_i^+$, and $p_i^-$. Edge $s_i$ refers to the \emph{switch} edge.
For each clause $C_j=(\ell_j^1\vee \ell_j^2\vee \ell_j^3)$ and each $t\in[3]$, create three parallel edges between $a(\ell_j^t)$ and $c_j$, denoted by $g_j^t$, $q_j^t$, and $r_j^t$. Edge $g_j^t$ refers to the \emph{main occurrence} edge.
Hence, the constructed multigraph has $3n+9m$ edges.

It remains to define the valuations. In the variable gadget for $x_i$, both
$x_i$ and $\bar{x}_i$ value $s_i$ at $4$; moreover, $x_i$ values $p_i^+$ at $1$ and $p_i^-$ at $0$, while $\bar{x}_i$ values $p_i^+$ at $0$ and $p_i^-$ at $1$.
For the gadget corresponding to the literal $\ell_j^t$ in the clause $C_j$, the literal agent $a(\ell_j^t)$ values $g_j^t,q_j^t,r_j^t$ at $1,0,1$, respectively, while the clause agent $c_j$ values them at $2,1,0$, respectively. 
All valuations are additive, and all values belong to $\{0,1,2,4\}$.
Figure~\ref{fig:propM-hardness-gadgets} illustrates the gadgets and valuations.
Each literal agent has total value $9$: value $4$ from the switch edge, value $1$ from exactly one of $p_i^+$ and $p_i^-$, and value $2$ from each of its two literal occurrences. Each clause agent also has a total value $9$, since each of its three occurrence gadgets contributes value $2+1$. 
Therefore $\mathsf{PROP}_i = \frac{9}{2}$ for every agent $i$.

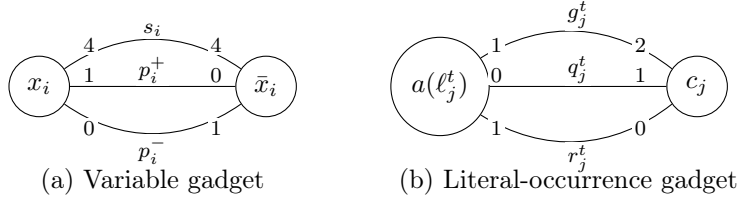
\begin{figure}[t]
\centering
\begin{tikzpicture}[
    agent/.style={circle,draw,minimum size=8mm,inner sep=1pt,font=\small},
    litagent/.style={circle,draw,minimum size=13mm,inner sep=1pt,font=\small},
    val/.style={font=\scriptsize,inner sep=1pt,fill=white},
    elab/.style={font=\scriptsize,inner sep=1pt,fill=white},
    panel/.style={font=\small}
]

\node[agent] (xi) at (0,0) {$x_i$};
\node[agent] (nxi) at (3,0) {$\bar{x}_i$};

\draw (xi) to[bend left=35] node[elab,above] {$s_i$} (nxi);
\node[val] at ($(xi)!0.22!(nxi)+(0,0.55)$) {$4$};
\node[val] at ($(xi)!0.78!(nxi)+(0,0.55)$) {$4$};

\draw (xi) -- node[elab,above] {$p_i^+$} (nxi);
\node[val] at ($(xi)!0.22!(nxi)+(0,0.15)$) {$1$};
\node[val] at ($(xi)!0.78!(nxi)+(0,0.15)$) {$0$};

\draw (xi) to[bend right=35] node[elab,below] {$p_i^-$} (nxi);
\node[val] at ($(xi)!0.22!(nxi)+(0,-0.55)$) {$0$};
\node[val] at ($(xi)!0.78!(nxi)+(0,-0.55)$) {$1$};

\node[panel] at (1.5,-1.25) {(a) Variable gadget};

\node[litagent] (lit) at (5.3,0) {$a(\ell_j^t)$};
\node[agent] (cj) at (8.7,0) {$c_j$};

\draw (lit) to[bend left=35] node[elab,above] {$g_j^t$} (cj);
\node[val] at ($(lit)!0.22!(cj)+(0,0.55)$) {$1$};
\node[val] at ($(lit)!0.78!(cj)+(0,0.55)$) {$2$};

\draw (lit) -- node[elab,above] {$q_j^t$} (cj);
\node[val] at ($(lit)!0.22!(cj)+(0,0.15)$) {$0$};
\node[val] at ($(lit)!0.78!(cj)+(0,0.15)$) {$1$};

\draw (lit) to[bend right=35] node[elab,below] {$r_j^t$} (cj);
\node[val] at ($(lit)!0.22!(cj)+(0,-0.55)$) {$1$};
\node[val] at ($(lit)!0.78!(cj)+(0,-0.55)$) {$0$};

\node[panel] at (7.0,-1.25) {(b) Literal-occurrence gadget};

\end{tikzpicture}
\caption{The gadgets used in the reduction. 
The label closer to an endpoint is the value of the corresponding edge for that endpoint.
For each clause $C_j=(\ell_j^1\vee \ell_j^2\vee \ell_j^3)$, the construction creates three copies of the literal-occurrence gadget in panel (b), one for each $t\in[3]$, all sharing the clause agent $c_j$.}
\label{fig:propM-hardness-gadgets}
\end{figure}

We first present a simple consequence of PO. 
In any PO orientation, every edge with exactly one endpoint that has a positive value on it must be assigned to that endpoint.
Hence, for each $i\in [n]$, $p_i^+$ is assigned to $x_i$, $p_i^-$ is assigned to $\bar{x}_i$, and, for every $j\in[m]$ and $t\in[3]$, $q_j^t$ is assigned to $c_j$, and $r_j^t$ is assigned to $a(\ell_j^t)$.
Otherwise, moving such an edge from the endpoint with a value of zero to the other endpoint with a positive value would strictly improve one agent without hurting others.

Suppose first that the \textsc{2P2N-3SAT} instance admits a satisfying assignment.
We construct an orientation $\pi$ as follows. For each $i\in [n]$, if $x_i$ is \texttt{True}, assign $s_i$ to $x_i$; otherwise, assign $s_i$ to $\bar{x}_i$. 
Assign all edges $p_i^+$, $p_i^-$, $q_j^t$, and $r_j^t$ to their unique positive-valuing endpoints. 
Finally, for every $j\in [m]$ and $t\in [3]$, assign $g_j^t$ to $c_j$ if $\ell_j^t$ is \texttt{True} and to $a(\ell_j^t)$ otherwise.

We claim that $\pi$ is PROPm. A literal agent receiving its switch edge has
value $7$, coming from the switch edge, its valued $p$-edge, and its two $r$-edges. 
A literal agent not receiving its switch edge corresponds to a false literal and therefore receives both of its $g$-edges, so its value is $5$. 
Thus, every literal agent has value at least $\frac{9}{2}$. Each clause agent $c_j$ receives its three $q$-edges, with a total value of $3$. 
Since the truth assignment satisfies $C_j$, at least one literal in $C_j$ is \texttt{True}, so $c_j$ also receives at least one $g$-edge of value $2$. 
Hence, $v_{c_j}(\pi_{c_j}) \geq 5 > \frac{9}{2}$, and $\pi$ is PROPm.

We next show that $\pi$ is PO. 
Suppose, for contradiction, that an orientation $\pi'$ Pareto dominates $\pi$.
By the observation above, we may assume without loss of generality that $\pi'$ agrees with $\pi$ on all edges $p_i^+,p_i^-,q_j^t,r_j^t$. 
Fix a variable $x_i$, and let $T_i$ be the literal
agent receiving $s_i$ in $\pi$ and $F_i$ the other literal agent. Agent $T_i$
has value $7$ in $\pi$. If $s_i$ were assigned to $F_i$ in $\pi'$, then even if
$T_i$ received all its other positive-valued incident edges, its value would be
at most $5$, contradicting that $\pi'$ Pareto dominates $\pi$.
Thus, every switch edge is assigned
in $\pi'$ as in $\pi$.

Now consider $F_i$. Since $F_i$ still does not receive the switch edge and
since in $\pi$ it already receives all other incident edges that it values positively, it cannot lose either of its main occurrence edges in $\pi'$.
Therefore, every main occurrence edge corresponding to a false literal remains
assigned to its literal agent. Consequently, a clause agent cannot gain any such
edge. If a clause agent lost a main occurrence edge that it receives in $\pi$,
its value would strictly decrease and could not be compensated. Hence, every main
occurrence edge is also assigned in $\pi'$ as in $\pi$. 
Thus $\pi'=\pi$, which contradicts the existence of a strict Pareto improvement. 
Therefore, $\pi$ is Pareto optimal.

Conversely, suppose that there exists an orientation $\pi$ that is both PROPm and Pareto optimal. 
By the observation above, all edges $p_i^+$, $p_i^-$, $q_j^t$, and $r_j^t$ are assigned to their unique positive-valuing endpoints.
Define a truth assignment by setting $x_i$ to \texttt{True} if and only if $s_i$ is assigned to $x_i$.

We prove that every clause is satisfied. Consider a literal agent that does not receive its switch edge, and let $k\in\{0,1,2\}$ be the number of its main occurrence edges that it receives in $\pi$. 
Its value is $3+k$: value $1$ from its valued $p$-edge, value $2$ from its two $r$-edges, and value $k$ from its main occurrence edges. 
Its PROPm correction term, in $\pi$, is $0$. Indeed, the opposite literal agent receives an edge of value $0$ from this agent's perspective, and each adjacent clause agent receives the corresponding $q$-edge, which also has value $0$ for the literal agent. Since $\pi$ is PROPm, we have $3+k\ge \frac{9}{2}$, and hence $k=2$. 
Thus, every literal agent that does not receive its switch edge must receive both of its main occurrence edges.

Now consider a clause agent $c_j$, and let $k$ be the number of main occurrence edges among $g_j^1,g_j^2,g_j^3$ assigned to it. Agent $c_j$ has a value of $3+2k$, because it receives all three $q$-edges. Its PROPm correction term is also $0$:
for every adjacent literal agent, the corresponding $r$-edge is assigned to the literal agent and has value $0$ for $c_j$. Since $\pi$ is PROPm, $3+2k\ge \frac{9}{2}$, so $k\ge 1$. Hence, every clause agent receives at least one main occurrence edge.
Let $g_j^t$ be a main occurrence edge assigned to $c_j$. Then the literal agent $a(\ell_j^t)$ does not receive one of its two main occurrence edges. 
By the previous paragraph about literal agents, $a(\ell_j^t)$ must receive its switch edge. 
By the definition of the truth assignment, the literal $\ell_j^t$ is true. 
Thus $C_j$ is satisfied.

Therefore, the \textsc{2P2N-3SAT} instance is satisfiable if and only if the constructed instance admits an orientation that is both PROPm and Pareto optimal. 
The construction is polynomial, creates $2n+m$ agents and $3n+9m$ edges, and uses only values in $\{0,1,2,4\}$. The decision problem is therefore NP-hard.
\end{proof}

\section{Missing Proofs in \Cref{sec:PROPavg}}

\subsection{Proof of \Cref{prop:prop-avgNotExist}} \label{sec:prop-avgNotExistAppendix}


\begin{proof}
Consider the two $K_6$ gadgets in Figure~\ref{fig:PROPavg}. The bridge has value $1$ at both roots. Each spoke has value $0$ at its root and $1$ at its leaf, and all leaf--leaf edges have value $0$ at both endpoints. Every agent has proportional share $\frac{1}{2}$.
In any orientation, one root, say $1$, loses the bridge. Its value is then zero, and its correction is $\frac{1}{1+z}$, where $z$ is the number of spokes it loses. Thus $z\le1$, so at least four leaves lose their spoke to root $1$. Each such leaf has value zero and must lose at most one of its four zero-valued leaf--leaf edges. It therefore receives at least three edges inside the $K_5$ on the leaves. These four leaves require at least $12$ edges, but the $K_5$ has only $10$, a contradiction.
\end{proof}

\end{document}